\documentclass[a4paper,11pt]{article}
\usepackage{amsmath}

\usepackage[small]{titlesec}
\titleformat*{\section}{\bfseries}
\titleformat{\subsection}[runin]{\bfseries\itshape}{\thesubsection.}{3pt}{\space}[.]
\titleformat{\subsubsection}[runin]{\itshape}{\thesubsection.}{3pt}{\space}[.]

\usepackage[labelfont=sc]{caption}

\usepackage{ stmaryrd }

\usepackage[long, nolistings, short]{optional} %

\usepackage{ifthen}

\usepackage[english]{babel}
\usepackage{color}
\usepackage{amsmath, latexsym, amsthm, amsfonts,bm,amssymb,mathscinet} 
\usepackage{natbib}
\usepackage[a4paper, text={5.8in,8.5in},centering]{geometry}

\usepackage{verbatim}
\usepackage{fancyvrb}
\usepackage{wrapfig}

\usepackage[utf8]{inputenc}
\usepackage[dvipsnames]{xcolor}
\usepackage[colorlinks,linkcolor=blue, citecolor=blue, urlcolor=blue]{hyperref}

\usepackage[normalem]{ulem}
\usepackage{enumitem}
\usepackage{bbm}

\usepackage{booktabs}

\newcommand{\forw}[1]{\cF_{ #1}}

\newcommand{\backw}[1]{\cB_{#1}}

\DeclareMathSymbol{\shortminus}{\mathbin}{AMSa}{"39}

\newcommand{\veryshortarrow}[1][3pt]{\mathrel{%
   \hbox{\rule[\dimexpr\fontdimen22\textfont2-.2pt\relax]{#1}{.4pt}}%
   \mkern-4mu\hbox{\usefont{U}{lasy}{m}{n}\symbol{41}}}}



\makeatletter
\newcommand*\bigcdot{\mathpalette\bigcdot@{.5}}
\newcommand*\bigcdot@[2]{\mathbin{\vcenter{\hbox{\scalebox{#2}{$\m@th#1\bullet$}}}}}
\makeatother
\newcommand{\pf}{\veryshortarrow}

\usepackage{tikz-cd}
\usetikzlibrary{circuits.logic.US}
\usetikzlibrary{positioning}
\usetikzlibrary{fit}
\usetikzlibrary{cd}
\usetikzlibrary{arrows}
\usetikzlibrary{calc}
\usetikzlibrary{decorations.markings}
\usetikzlibrary{shapes.geometric}
\tikzset{ed/.style={auto,inner sep=2pt,font=\scriptsize}} %
\tikzset{>=stealth}

\tikzset{vert/.style={draw,circle, minimum size=6mm, inner sep=0pt, fill=white}}
\tikzset{vertblank/.style={ minimum size=6mm, inner sep=0pt, fill=white}}

\tikzset{vertbig/.style={draw,circle, minimum size=8mm, inner sep=0pt, fill=white}}
\tikzset{->-/.style={decoration={
      markings,
      mark=at position #1 with {\arrow{>}}},postaction={decorate}}}

\tikzset{edge/.style={line width=0.5pt, decoration={markings,mark=at position 1 with %
    {\arrow[scale=1.5,>=stealth]{>}}},postaction={decorate}}}

\tikzset{dotted/.style={black!30, line width=0.5pt}}

\pgfdeclarelayer{edgelayer}
\pgfdeclarelayer{nodelayer}
\pgfsetlayers{edgelayer,nodelayer,main}

\tikzstyle{morphism}=[fill=white, draw=black, shape=rectangle]
\tikzstyle{medium box}=[fill=white, draw=black, shape=rectangle, minimum width=0.8cm, minimum height=0.9cm]
\tikzstyle{large morphism}=[fill=white, draw=black, shape=rectangle, minimum width=1.7cm, minimum height=1cm]
\tikzstyle{bn}=[fill=black, draw=black, shape=circle, inner sep=1.5pt]
\tikzstyle{effect}=[fill=white, draw=black, regular polygon, regular polygon sides=3, minimum width=0.4cm, inner sep=0pt]
\tikzstyle{state}=[fill=white, draw=black, regular polygon, regular polygon sides=3, minimum width=0.4cm, shape border rotate=180, inner sep=0pt]
\tikzstyle{medium state}=[fill=white, draw=black, regular polygon, regular polygon sides=3, minimum width=1.3cm, inner sep=0pt, shape border rotate=180]
\tikzstyle{large state}=[fill=white, draw=black, regular polygon, regular polygon sides=3, minimum width=2.2cm, shape border rotate=180, inner sep=0pt]
\tikzstyle{wn}=[fill=white, draw=black, shape=circle, inner sep=1.5pt]

\tikzstyle{arrow}=[->]
\tikzstyle{dashed box}=[-, dashed]

\tikzset{none/.style={%
     append after command={%
       \pgfextra{\node [right] at (\tikzlastnode.mid east) {{\tiny\tikzlastnode}};}
     }}}
\tikzstyle{none}=[]

\newcounter{snippet}
\opt{nolistings}{

\newenvironment{snippet}{\opt{long}{}\VerbatimEnvironment\begin{Verbatim}}{\end{Verbatim}\addtocounter{snippet}{1}\color{black}}

}
    
\opt{listings}{
\usepackage{minted}

\newenvironment{snippet}{
    \VerbatimEnvironment\begin{minted}[frame=lines,framesep=2mm,fontsize=\footnotesize, xleftmargin=0.5em, mathescape, linenos, label=Snippet \arabic{snippet}, labelposition=bottomline, escapeinside=||, mathescape=true]{julia}}{\end{minted}\addtocounter{snippet}{1}}

}

\DeclareSymbolFont{bbsymbol}{U}{bbold}{m}{n}
\DeclareMathSymbol{\bbsemi}{\mathbin}{bbsymbol}{"3B}
\DeclareMathSymbol{\bbcomma}{\mathbin}{bbsymbol}{"2C}
\newcommand{\comp}{}
\usepackage{newunicodechar}
\newunicodechar{ä}{\"a}
\newunicodechar{ü}{\"u}
\newunicodechar{ö}{\"o}
\newunicodechar{Ä}{\"A}
\newunicodechar{Ü}{\"U}
\newunicodechar{Ö}{\"O}
\newunicodechar{ß}{\ss}
\newunicodechar{λ}{\ensuremath{\lambda}}
\newunicodechar{θ}{\ensuremath{\theta}}
\newunicodechar{σ}{\ensuremath{\sigma}}
\newunicodechar{ν}{\ensuremath{\nu}}
\newunicodechar{μ}{\ensuremath{\mu}}
\newunicodechar{Σ}{\ensuremath{\Sigma}}
\newunicodechar{ˣ}{\ensuremath{^*}}
\newunicodechar{β}{\ensuremath{\beta}}
\newunicodechar{′}{\ensuremath{^\prime}}
\newunicodechar{∘}{\ensuremath{\circ}}
\newunicodechar{Δ}{\ensuremath{\Delta}}
\newunicodechar{ᵒ}{\ensuremath{^\circ}} 
\newunicodechar{Φ}{\ensuremath{\Phi}} 

\newcommand{\Dup}{{\large\lhd}}

\newcommand{\dd}{{\,\mathrm d}}

\renewcommand{\th}{\theta}

\renewcommand{\phi}{\varphi}

\renewcommand{\scr}[1]{{\mathcal #1}}
\newcommand{\argmin}{\operatornamewithlimits{argmin}}

\newcommand{\EE}{\mathbb{E}}

\newcommand{\PP}{\mathbb{P}}
\newcommand{\ind}{\mathbf{1}}
\newcommand{\RR}{\mathbb{R}}

\newcommand{\bem}{\begin{bmatrix}}
\newcommand{\enm}{\end{bmatrix}}
\newcommand{\bs}[1]{{\boldsymbol #1}}
\newcommand{\T}{{\prime}}

\newcommand{\id}{\operatorname{id}}

\renewcommand{\P}{\mathbb{P}} %
\newcommand{\Q}{\mathbb{Q}} %

\theoremstyle{definition}
\newtheorem{thm}{Theorem}[section]
\newtheorem{prop}[thm]{Proposition}
\newtheorem{lem}[thm]{Lemma}
\newtheorem{cor}[thm]{Corollary}
\newtheorem{rem}[thm]{Remark}
\newtheorem{ex}[thm]{Example}
\newtheorem{defn}[thm]{Definition}

\newcommand{\E}{\mathbb{E}}

\newcommand{\on}[1]{\operatorname{#1}}

\newcommand{\Bm}{\begin{bmatrix}}
\newcommand{\Em}{\end{bmatrix}}

\newcommand{\hedge}[3]{h_{#1\pf#2}(#3)}

\newcommand{\hedges}[2]{h_{\pf#1}(#2)}
\newcommand{\hvertex}[2]{h_{#1}(#2)}

\newcommand{\pedge}[4]{\kappa_{#1 \pf#2}(#3; #4)}

\newcommand{\pstaredge}[4]{\kappa^\star_{#1\pf#2}(#3; #4)}
\newcommand{\pcircedge}[4]{\kappa^\circ_{#1\pf#2}(#3; #4)}

\newcommand{\optic}[2]{\langle \cF_{#1} \mid \cB_{#2}\rangle}

\newcommand{\ev}{\operatorname{ev}}

\newcommand{\weight}[3]{w_{#1\pf#2}(#3)}

\newcommand{\pder}[2][]{\frac{\partial #1}{\partial #2}}
\makeatletter

\def\lstAZ{A, B, C, D, E, F, G, H, I, J, K, L, M, N, O, P, Q, R, S, T, U, V, W, X, Y, Z}
\def\lstaz{a, b, c, d, e, f, g, h, i, j, k, l, m, n, o, p, q, r, s, t, u, v, w, x, y, z}

\def\lstAZBB{B, C, D, E, F, G, H, I, J, K, L, M, N, O, P, Q, R, T, U, V, W, X, Y, Z}

\newcommand{\MkScr}[1]{\expandafter\def\csname s#1\endcsname{\mathscr{#1}}}
\newcommand{\MkUp}[1]{\expandafter\def\csname u#1\endcsname{\mathrm{#1}}}
\newcommand{\MkBold}[1]{\expandafter\def\csname b#1\endcsname{\mathbf{#1}}}
\newcommand{\MkFrak}[1]{\expandafter\def\csname f#1\endcsname{\mathfrak{#1}}}
\newcommand{\MkCal}[1]{\expandafter\def\csname c#1\endcsname{\mathcal{#1}}}
\newcommand{\MkBB}[1]{\expandafter\def\csname #1#1\endcsname{\mathbb{#1}}}

\@for\i:=\lstAZ\do{%
	\expandafter\MkScr \i  %
	\expandafter\MkFrak \i  %
	\expandafter\MkUp \i %
	\expandafter\MkBold \i %
	\expandafter\MkCal \i  %
}    
\@for\i:=\lstaz\do{%
	\expandafter\MkUp \i   }    

\@for\i:=\lstAZBB\do{%
	\expandafter\MkBB \i     }

\makeatother

\DeclareFontFamily{U}{matha}{\hyphenchar\font45}
\DeclareFontShape{U}{matha}{m}{n}{
      <5> <6> <7> <8> <9> <10> gen * matha
      <10.95> matha10 <12> <14.4> <17.28> <20.74> <24.88> matha12
      }{}
\DeclareSymbolFont{matha}{U}{matha}{m}{n}
\DeclareMathSymbol{\varrightharpoonup}{3}{matha}{"E1}
\DeclareMathSymbol{\varleftharpoonup}{3}{matha}{"E0}

\newcommand{\before}{\leadsto}
\renewcommand{\tilde}{\widetilde}

\DeclareMathOperator{\pa}{pa}
\DeclareMathOperator{\anc}{anc}
\DeclareMathOperator{\ch}{ch}

\newsavebox\bsbcopier
\savebox\bsbcopier{%
  \begin{tikzpicture}[baseline=0pt,line width=0.5pt]
    \node[bn,scale=0.7] (a) at (0, 2.8pt) {};
    \draw (a) -- +(-180:.21);
    \draw (a) -- +(-45:.21);
    \draw (a) -- +(45:.21);
  \end{tikzpicture}} 
\newsavebox\bsbcopierb
\savebox\bsbcopierb{%
  \begin{tikzpicture}[baseline=0pt,line width=0.5pt]
    \node[bn,scale=0.7] (a) at (0, 2.8pt) {};
    \draw (a) -- +(0:.21);
    \draw (a) -- +(-225:.21);
    \draw (a) -- +(-135:.21);
  \end{tikzpicture}}

\newsavebox\bsinver
\savebox\bsinver{%
  \begin{tikzpicture}[baseline=0pt,line width=0.5pt]
    \node[wn,scale=0.7] (a) at (0, 2.8pt) {};
    \draw (a) -- +(0:.21);
    \draw (a) -- +(-180:.21);
  \end{tikzpicture}}

\usepackage{mathtools}

\DeclarePairedDelimiterXPP\expec[1]{\E}[]{}{

#1}

\makeatletter
\renewcommand*{\@fnsymbol}[1]{\ensuremath{\ifcase#1\or *\or 
    \mathsection\or \mathparagraph\or \|\or **\or \dagger\dagger
    \or \ddagger\ddagger \else\@ctrerr\fi}}
\makeatother
\definecolor{darkgreen}{rgb}{0.1, 0.3, 0.23}

\date{}

\title{\bf Automatic Backward Filtering Forward Guiding for Markov processes and graphical models}

\author{\small \bf Frank van der Meulen\thanks{De Boelelaan 1111, 1081HV Amsterdam, The Netherlands }\\
\small Department of Mathematics \\
\small Vrije Universiteit Amsterdam\\
\small Email: f.h.van.der.meulen@vu.nl
\and \small \bf Moritz Schauer\thanks{ Chalmers Tvärgata 3, 41296 Göteborg, Sweden}\\
\small Chalmers University of Technology\\\small and University of Gothenburg\\
\small E-mail: smoritz@chalmers.se}

\begin{document}
\maketitle


\begin{abstract} 
\parindent0pt
\parskip1ex

We incorporate discrete and continuous time Markov processes as building blocks into probabilistic graphical models with latent and observed variables.  We introduce the automatic Backward Filtering Forward Guiding (BFFG) paradigm (\cite{mider2021continuous}) for programmable inference on latent states and model parameters.
Our starting point is a generative model, a forward description of the probabilistic process dynamics.
We backpropagate the information provided by observations through the model to transform the generative (forward) model into a pre-conditional model guided by the data. It approximates the actual conditional model with known likelihood-ratio between the two.
The backward filter and the forward change of measure are suitable to be incorporated into a probabilistic programming context because they can be formulated as a set of transformation rules.

The guided generative model can be incorporated in different approaches to efficiently sample latent states  and parameters conditional on observations. We show applicability in a variety of settings, including Markov chains with discrete state space, interacting particle systems, 
state space models, branching diffusions and Gamma processes.

 \end{abstract}



\numberwithin{equation}{section}
\sloppy

\section{Introduction}
\subsection{Problem description}

A large number  of scientific disciplines makes use of stochastic processes for probabilistic modelling of time evolving phenomena and complex systems under uncertainty.  In this setting data assimilation and statistical inference is challenging. More so, if this is to be incorporated as components into larger probabilistic models, to be described using modelling languages, and if it is desired to perform inference in an at least semi-automated fashion in probabilistic programming.  In this paper we provide a way to incorporate Markovian stochastic processes as building blocks into probabilistic graphical models and efficiently perform inference on them. 

Probabilistic graphical models capture the structure of a probabilistic model in the form of a graph from which conditional independence properties can easily be extracted. We assume that the dependence structure of random quantities  is represented by a directed acyclic graph (DAG) with vertices $t \in \cT$ and edges directed from a set of source vertices towards the leaves  (commonly used  terminology for graphical models is reviewed in appendix \ref{sec:notation}). 
To each vertex $t$ of the DAG corresponds a random quantity $X_t$, constituting  a Bayesian network. 
Often the DAG arises as the dependency structure of a probabilistic program.
The conditional distribution of $X_t$  given all  $X_s$ corresponding to the ancestors  $s$ of $t$ in the DAG is prescribed and assumed to depend only on the parents, the random quantity $X_{\pa(t)}$, the set of vertices $s$ for which there is a directed edge from $s$ to $t$.
See  Figure~\ref{fig:tree} for a simple example with a single source (root). 

Assume we only observe $X_v$ for each leaf vertex  $v$ of the DAG. Leaves can be attached anywhere in the DAG, also at the sources. Therefore edges ending in leaves can be conceptualised as observation schemes, possibly  representing random errors. The smoothing problem consists of finding the exact joint distribution of all $X_s$, where $s \in \cT$ is not a leaf, conditional on the observations. In this paper we give a novel way to sample from the smoothing distribution  thereby extending well known algorithms for graphical models.  Contrary to other methods, the proposed framework does not require modification of the DAG or dependency structure of the probabilistic program, for example by addition of edges.
We are especially interested in models where some $X_t$ are obtained as the endpoint of the evolution of a Markov process starting from the parent $X_s$ over a certain time span.  The transitions may depend on parameters in the Markov kernel  which need to be estimated. 

The problem setting encompasses training of stochastic residual neural networks \citep{wang2018resnets} and their continuous-time counterparts  neural stochastic differential equations \citep{chen2018neural}, as well as other settings where the output of a  complicated latent process is observed  and the latent process is obtained from iteration and composition of models which are tractable (only) locally in time, in space or along single edges in the DAG.

\subsection{Approach}

We start from a Markovian forward description of the possibly nonlinear,  non-reversible process dynamics. In essence, by first  filtering backwards an approximate process from the leaves towards the root and then performing a change of measure on the transitions of the generative (forward) model using the filtering information, we derive a generative model \emph{guided}  by the data. This model provides a structure preserving  approximation to the true model conditional on the  observations to sample from.  Notably, we also transform transitions which correspond to continuous time Markov processes using an exponential change of measure \citep{PalmowskiRolski2002}. 
This \emph{pre-conditional} model, which combines information from data and the model, can then be incorporated in different approaches to efficiently sample the actual smoothing distribution, for example Hamiltonian Monte Carlo, particle methods or variational inference. 
The setup allows to obtain a differentiable likelihood through a reparametrisation in a natural way.

As a familiar special case of the procedure we obtain the backwards filtering forward sampling algorithm for  linear state space models (in this case the guided approximation is exact and can be directly sampled as latent path given the observations). Also our previously introduced backward filtering forward guiding (BFFG) algorithm for discretely observed stochastic differential equations \citep{mider2021continuous} is a special case.  More examples, involving Gamma processes, continuous-time Markov chains and Poisson processes are also covered in this paper. Key to this is the insight that the dynamics of a conditioned Markov process or graphical model can be obtained by Doob's $h$-transform and that guided proposals (as defined in \cite{schauer2017guided})  result from an approximation to this transform.   The idea of using an approximate $h$-transform dates back to at least \cite{Glynn1989} in the setting of rare event simulation.  

The backward filter and the forward change of measure are suitable to be incorporated into probabilistic programming environments. 
We use arguments from category theory to derive a principle of compositionality for our method, that is a principle which ``describes and quantifies how complex things can be assembled out of simpler parts'', as the eponymous journal Compositionality defines it.  Very related properties of Bayesian updates are considered in \cite{smithe2020bayesian}.
By this it is enough to provide a library of transformation rules and tractable approximations for backward filtering each single transition, the 
elementary building blocks of a model, and an implementable composition rule to obtain guided processes for larger models.
Note that we speak of ``guided processes'' rather than ``guided proposals'' (used in earlier works), the motivation being that its applicability is not restricted to a proposal distribution in a Metropolis-Hastings algorithm.

We provide a library of some of such rules and approximations as Julia package \texttt{Mitosis.jl},\footnote{\url{https://github.com/mschauer/Mitosis.jl}.} following the spirit of the \texttt{ChainRules.jl} package, \cite{ChainRules.jl}, providing common differentiation rules to the different automatic differentiation packages in Julia ecosystem.

\subsection{Related work}
 
There is a vast literature on filtering and smoothing on graphical models and continuous-time Markov models. Much attention has been paid to state-space models, especially in particle filtering. The proposed framework in this paper is agnostic with respect to the choice of inferential method and fits within the class of message passing algorithms. In Section \ref{sec:related work} we provide a detailed overview of references to related work. The companion paper \cite{bffg_introduction} contains a non-technical introduction to this work, targeted at a broader  audience.

\subsection{Novelty and contribution}

Building up on successes with guided processes  \citep{schauer2017guided} in the specific, but challenging problem of likelihood-based inference for diffusion process, this work gives a new, holistic perspective on automatic probabilistic programming in general. 

Our approach stands out by being \emph{generally applicable} to both discrete time as continuous-time
processes, discrete and continuous state spaces,  \emph{agnostic} with respect to the choice of inferential method,
\emph{automatic}, being suitable for automatic transformation of probabilistic programs,
and \emph{compositional}, the transformation of a probabilistic program relies on the composition of the transformations of its elementary building blocks in a backward and a forward pass.

\subsection{Outline}

In Subsection \ref{subsec:probgraphcmodels} we explain the setting and some notation.  In Subsection \ref{subsec: doob} we 
  define Doob's $h$-transform for Markov processes indexed on  a tree. 
  In Subsection \ref{subsec:guidedproc} we introduce guided processes for directed trees and, more generally, Bayesian  networks. In Subsection \ref{sec:continuoustime} we explain how continuous time processes, evolving over an edge, can be incorporated. 
Subsequently,  the computational structure of Doob's $h$-transforms and guided processes is detailed in Section \ref{sec:compositionality} for the specific case of a line graph. Here, we explain how an efficient implementation is obtained by piecing together forward and backward maps.
For the more general setting of a directed acyclic graph we explain how the  compositional structure can be unpacked in Section \ref{sec:unpacking}.
In Section \ref{sec:comp_bffg} we prove compositionality results for the pairing of forward- and backward maps on a directed tree. Examples  of guided processes with tractable backward map are given in Section \ref{sec:examples}.  Examples for continuous-time processes are given in Section \ref{sec:examples_continous}.  Various ways in which the obtained results can be used for inference are discussed in Section \ref{sec:parestimation}.  

 Appendix \ref{sec:related work} contains a summary of related work.  
 Some notation for graphical models  is reviewed in Appendix \ref{sec:notation}. Appendix \ref{sec:proofs} contains some postponed proofs.  

\begin{figure}
\begin{center}
\includegraphics[scale=0.3]{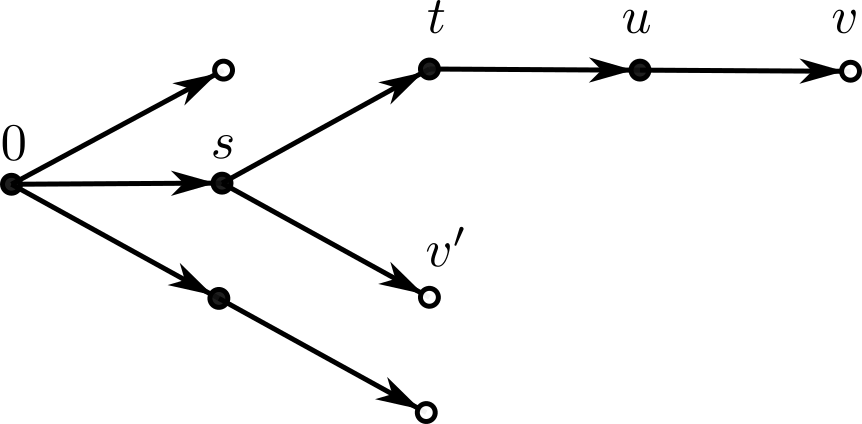}
\end{center}
\caption{A simple tree,  with $\bullet$ denoting  latent vertices and  $\circ$ leaf-vertices. The root is denoted $0$. Along each edge the process evolves according to either one step of a discrete-time Markov chain or a time-span of a continuous-time Markov process. }
\label{fig:tree}
\end{figure}

\section{Doob's $h$-transform and guided processes}\label{sec:graphmodels}

We first recap some elementary definitions on Markov kernels, as these are of key importance in all that follows. 
\subsection{Markov kernels}
Let $S=(E,\fB)$ and $S'=(E',\fB')$ be 
Borel measurable spaces. A Markov kernel between $S$ and $S'$ is denoted by $\kappa\colon S \rightarrowtriangle S'$  (note the special type of arrow), where $S$ is the ``source''  and $S'$ the ``target''. That is, $\kappa\colon E \times \fB' \to [0,1]$, where {\it (i)} for fixed $B\in \fB'$ the map $x\mapsto \kappa(x, B)$ is $\fB$-measurable and {\it (ii)} for fixed $x\in E$, the map $B\mapsto \kappa(x, B)$ is a probability measure on $S'$. 
On a measurable space $S$, denote the sets of bounded measures and bounded measurable functions (equipped with the supremum norm)  by $\cM(S)$ and $\mathbf{B}(S)$ respectively. 
The kernel $\kappa$ induces a {\it pushforward} of the measure $\mu$ on $S$ to $S'$   via 
\begin{equation}\label{eq: pushforward}
\mu \kappa(\cdot)  = \int_{E} \kappa(x, \cdot) \mu(\!\dd x), \qquad \mu \in \cM(S).
\end{equation}
The linear continuous operator $\kappa\colon \mathbf B(S') \rightarrow \mathbf B(S)$ is defined by 
\begin{equation}\label{eq: pullback}
\kappa h(\cdot) = \int_{E} h(y) \kappa(\cdot, \dd y),\qquad  h \in \mathbf{B}(S').
\end{equation} 
We will refer to this operation as the {\it pullback} of $h$ under the kernel $\kappa$. 
Markov kernels $\kappa_1\colon S_0 \to S_1 $ and $\kappa_2\colon S_1 \to S_2$  
can be composed  by the Chapman-Kolmogorov equations, here written as juxtaposition
\begin{equation}\label{eq:chapman}
(\kappa_1 \comp \kappa_2)(x_0, \cdot) = \int_{E_1}  \kappa_2(x_1, \cdot) \kappa_1(x_0, \dd x_1),\qquad x_0 \in E_0.
\end{equation}
The unit $\id$ for this composition is the  identity function considered as a Markov kernel, $\id(x, \dd y) = \delta_{x}(\!\dd y)$.
The product kernel $\kappa \otimes \kappa'$ on $S \otimes S'$
 is defined on the cylinder sets by $
(\kappa  \otimes \kappa') ((x, x'), B \times B') = \kappa (x, B) \kappa' ( x', B'), $  	(where $x\in E$, $x' \in E'$, $B \in \fB$, $B' \in \fB'$) 
and then extended to a kernel on the product measure space.

\subsection{Probabilistic graphical models}\label{subsec:probgraphcmodels}
Let $\cG = (\cT, \cE)$ be a directed, finite  graph without parallel edges and self-loops, with vertices  $\cT$ and edges $\cE$, where $(s, t)$ or $s\pf t$ denotes  a directed edge from $s$ to $t$, $s, t \in \cT$. Without first argument, $\pf\!t$ denotes the edge or the set of edges directed to $t \in \cT$.
We write $s \before t$  if  $s = t$ or there is a directed path from $s$ to $t$. 
We denote by $\pa(t)$, $\ch(t)$, $\anc(t)$ the parents, children and ancestors of a vertex $t$ respectively. The set of leaves (sinks) is denoted by $\cV$. For each vertex $t$, we denote by $\cV_t = \{v \in \cV, t \before v\}$ the leaves which are descendants of $t$ (including $t$, if $t$ is a leaf).  Terms and notations from graph theory used here (see e.g.\ chapter 10 in \cite{Murphy2012}) are formally defined in Appendix \ref{sec:notation}.

To simplify the presentation we will assume a designated root vertex $0 \notin \cT$ with deterministic $X_0 = x_0$ and formally define $0 = \pa(s)$ for sources, that is vertices $s \in \cT$ without parent(s) in $\cT$. In this sense, either $x_0$ represents a fixed or $X_{\ch(0)}$ a random starting configuration with possibly multiple source vertices (sources in $\cS$).

Finally, define  $\cS = \cT\setminus \cV$ (the set of non-leaf vertices) and $\pa(\cV)$ to be the set of vertices in $\cT$ that are parent of at least one leaf node and set $\cS_0 = \cS \cup \{0\}$.

{\bf Setting: } {\it We consider a probabilistic graphical model on $\cG$, a random process  $(X_{t},\, t \in \cT)$  such that $X_t$ depends on its ancestors $\anc(t)$ only through its parent(s) $\pa(t)$
 and where random variables $X_s, X_t$, $s \not\before t$, $t \not\before s$ are independent conditional on all of their parents.
We  assume that $X_t$ takes values in Borel spaces $S_t = (E_t, \fB_t)$ 
satisfying
\[
\kappa_{\pa(t)\pf t}(x; \dd y) = \P(X_t \in \dd y \mid X_{\pa(t)} = x) 
\]
for given Markov transition kernels $\kappa_{\pa(t)\pf t}$ 
 (where such notation here and in the following should be read as $\int_A \kappa_{\pa(t)\pf t}(x; \dd y) = \P(X_t \in \dd A \mid X_{\pa(t)} = x)$ for all measurable $A$.)  
 }
 For $v \in \cV$ we assume the existence of transition densities $p_{\pa(v)\pf v}$  with respect to a dominating measures $\lambda$
\begin{equation}\label{eq:density_leaves}	
 \kappa_{\pa(v)\pf v}(x; \dd y)  =p_{\pa(v)\pf v}(x; y) \lambda(\!\dd y),\qquad v\in \cV, \quad x \in E_{\pa(v)} 
\end{equation}
Let $\P^\star$ denote the law of $X_\cS$ conditional on $X_\cV=x_\cV$. By Bayes theorem
\begin{equation}\label{eq:evidence_}
 \frac{\dd \P^\star}{\dd \P}(X_\cS) = \frac{\prod_{v\in \cV} p_{\pa(v)\pf v}(X_{\pa(v)}, x_v)}{\ev(x_0)}. \end{equation}
where $\ev(x_0)$ is known as  the evidence. 
If also the remaining transition kernels have a transition density $p$ with respect to a measure $\lambda$, i.e.~$\kappa_{\pa(t)\pf t}(x; \dd y) = p_{\pa(t)\pf t}(x, y) \lambda(\dd y)$, the joint density factorises as
\begin{equation}
{p(x_\cT)} = %
\prod_{t \in \mathcal T}  p_{\pa(t)\pf t}(x_{\pa(t)}; x_t).
\end{equation}

\subsection{Doob's $h$-transform  on a directed tree}\label{subsec: doob}
Doob's $h$-transform characterises the conditional distribution of a Markov process given information at the end of a time horizon/life time in a succinct, but not always tractable form. It generalises to the setting where  $\cG$ is a connected tree directed to the leaves, i.e.\   each vertex has a unique parent, except from the formal root vertex $0$. Thus, edges point towards the leaves and ``backwards'' is to be understood as pointing from leaves to the root.  

We are interested in the distribution  of $X_{\cS}$ conditional on $X_{\cV}$, i.e.\ the smoothing distribution. The values at the leaf vertices can be interpreted as observations without error. Alternatively,   $X_{\pa(\cV)}$ can be thought of as observed quantities with the  edges $\pa(v)\pf v$, $v \in \cV$  representing observation errors. 
 Note that each vertex, including the root $0$, may have one or more leaves as its children.  A simple example is depicted in Figure~\ref{fig:tree}.

With this notation, the process $X$ that is conditioned on  its values on the leaves has transition kernels
\[ \P(X_s \in \dd x_s \mid X_{\pa(s)} = x_{\pa(s)}, X_{\cV_s}=x_{\cV_s}) = \kappa^\star_{s}( x_{\pa(s)}; \dd x_s),\qquad s\in \cS \]
characterised by
\begin{equation}\label{eq:kappastar}
\kappa^\star_{s}( x_{\pa(s)}; \dd x_s) = \frac{ \kappa(x_{\pa(s)}, \dd x_s) h_s(x_s)}{\int \kappa(x_{\pa(s)}, \dd x_s) h_s(x_s) }, 
\end{equation}

where the   term on $h$ on the right-hand-side is the $h$-transform at vertex $t$, i.e.\ $h_t(x_t) = p(x_{\cV_t} \mid x_t)$ in Bayesian notation. 
Within the subtree with root vertex $t$, this can be interpreted as the likelihood of $x_t$ with respect to the measure $\lambda$.  Alternatively, if we would imply a flat prior on $x_t$, then $h_t(x_t)$ would be proportional to the posterior density of $x_t$ (assuming $x\mapsto h_t(x)$ to be integrable) where $x_t$ is the parameter and $\cV_t$ the observation.  
Note that the observations are dropped from the notation as they are fixed throughout (the same convention is used for example in \cite{cappe2005springer}; Section 3.1.4 in there). In case of hidden Markov models, a rigorous proof that the conditioned latent process is an inhomogeneous Markov chain is given in Section 3.3 of \cite{cappe2005springer} and the same argument can adapted to formally justify \eqref{eq:kappastar}. 

It is well known how $h_s$ can be computed recursively  starting from the leaves back to the root. These recursive relations have  reappeared in many papers, see for  instance \cite{Felsenstein1981},  \cite{briers2010smoothing}, \cite{Guarniero2017} and \cite{Heng2020}. This recursive computation is known as the {\it Backward Information Filter (BIF)}. 
Firstly, for any leaf vertex $v$ we define
\begin{equation}\label{eq:hleaf} h_{\pa(v) \pf v} (x)  := 
p_{\pa(v) \pf v}(x; x_{v})   \quad v \in \cV .\end{equation}
For other vertices $s$ we define 
\begin{equation}\label{eq:hparent} 
h_{\pa(s) \pf s} :=\kappa_{\pa(s) \pf s} h_s , \quad s \in \cS, 
\end{equation}
where $ \kappa_{\pa(s) \pf s} h_s $ is the pullback (as defined in Equation \eqref{eq: pullback}) 
 of $h_s$ under the kernel  $\kappa_{\pa(s) \pf s}$.
By the Markov property we have 
\begin{equation}\label{eq:split_tochilds} \hvertex{s}{x} = \prod_{t\in \ch(s)} \hedge{s}{t}{x}, \qquad s\in \cS_0.
\end{equation}
This can be interpreted as \emph{fusion}, collecting all messages from children at vertex $s$, the messages being $x\mapsto h_{s\pf t}(x)$. Then the dynamics of the conditional process is given in short
\begin{equation}\label{eq:def_pstar} \pstaredge{\pa(s)}{s}{x}{\dd y} = \pedge{\pa(s)}{s}{x}{\dd y} \frac{\hvertex{s}{y}}{\hedges{s}{x}} \end{equation}
and $\hedges{s}{x}$ shows up as normalising constant and one could define it that way. 
Note that equivalently $\kappa^\star$ is characterised  by the relation
\begin{equation}\label{eq: kappastar}
\frac{\kappa_{\pa(s) \pf s} f h_s }{\kappa_{\pa(s) \pf s} h_s} = \kappa^\star_{\pa(s) \to s} f,
\end{equation}
for bounded measurable test functions $f \in \bB(S_s)$.
If $h$ is the Doob $h$-transform, then (on a directed tree) the evidence defined in \eqref{eq:evidence_} satisfies  $\ev(x_0)=\hvertex{0}{x_0}$.

Let $X^\star$ be the process $X$ on the vertex set of $\cG$, starting in $x_0$ and conditoned on observations at $\cV$.
As $\cG$ is a tree, $X^\star$ is Markov process on $\cG$ at well  
and evolving according to the transition kernels in \eqref{eq:def_pstar}.

\subsection{Guided processes}\label{subsec:guidedproc}

The $h$-transform can only be computed in closed form in few cases: it assumes the graphical model to be a directed tree and even there, in general, it is intractable and so is $\kappa^\star$. 

Suppose it is possible to work with an approximation $g$ of $h$,  for example by choosing a Gaussian approximation to $X$, and then computing the $h$-transform for that approximation instead. Then one can use this $g$ to define an \emph{approximate}   $h$-transform, to be used as pre-conditional model which then serves as approximation of the conditional process and can be transformed into the actual conditional process via a change of measure, that is through sampling. 

Similar to the $h$-transform we define the \emph{guided process} $X^\circ$  by transforming each transition along an edge of the DAG by applying a change of measure to the forward kernel by specifying functions $g_{s\pf t}$ on all edges. 
Note that on a DAG, contrary to a directed tree, conditioning on the values at leaf-vertices changes the dependency structure (there is a conditional process $X^\star$ defined on the vertex set, but its dependency graph is different from that of $X$, in particular there is
no meaningful transition $\kappa^\star_{\pa(s)\pf s}$ for all $s$).  The guided process, unlike the conditioned process, is constructed to have the same dependency structure as the (unconditional) forward process. Crucially this means a sampler for the guided process has a similar complexity as a sampler for  $X$, predisposing guided process for sampling based inference.

\begin{defn}\label{defn: guided process} Let the maps $x \mapsto g_{s\pf t}(x)$ be specified for each edge $(s,t)$ in $\cT$ and define 
\begin{equation}\label{eq:split_tochilds_g} g_s(x)= \prod_{t\in \ch(s)} g_{s\pf t}(x), \qquad s\in \cS_0. \end{equation}
 We define the guided process $X^\circ$   as the process starting in $X^\circ_0 = x_0$
 and from the roots onwards evolving  \emph{on} the DAG $\cG$ according to  transition kernel
\[
\pcircedge{\pa(s)}{s}{x_{\pa(s)}}{\!\dd y}  =\frac{g_s(y)\pedge{\pa(s)}{s}{x_{\pa(s)}}{\!\dd y}}{\int g_s(y)\pedge{\pa(s)}{s}{x_{\pa(s)}}{\!\dd y}  }, \qquad s\in \cS.  \] \end{defn} 

Samples of $X^\circ$ can be transformed into samples of the conditional process. 
If $\P^\circ$ denotes the law of $X^\circ_\cS$, then it follows from the definition of $\kappa^\circ$ 
\[  \frac{\dd \P^\circ}{\dd \P}(X_\cS) = \prod_{s\in \cS} \left( \frac{g_s(X_s)}{\int g_s(y)\pedge{\pa(s)}{s}{X_{\pa(s)}}{\!\dd y}}\right).  \]
Recall the definition of $\P^\star$ in \eqref{eq:evidence_}. 
\begin{thm}\label{thm:lr_xstar_xcirc_dag}
If $\P^\star \ll \P^\circ$, then for $f \in \bB(S_\cS)$
\[\E \left[ f(X_{\cS}) \mid X_\cV= x_\cV\right]
 =\frac{g_{0}(x_{0})}{\ev(x_0)} \E  \left[ f(X^\circ_{\cS})  \prod_{s\in \cS}  \weight{\pa(s)}{s}{X^\circ_{\pa(s)}}  \prod_{v\in \cV} \frac{ h_{\pa(v) \pf v}(X^\circ_{\pa(v)})}{g_{\pa(v)\pf v}(X^\circ_{\pa(v)})} \right],
\] with weights defined by  
\begin{equation}\label{eq:dagweights}
\weight{\pa(s)}{s}{x_{\pa(s)}} = \frac{\int   g_s(y)\pedge{\pa(s)}{s}{x_{\pa(s)}}{\!\dd y}  }{\prod_{u \in \pa(s)} g_{u\pf s}(x_u) }\qquad s\in \cS .
\end{equation}

\end{thm}
\begin{proof}
By definition of the weights \[  \frac{\dd \P^\circ}{\dd \P}(X_\cS) = \prod_{s\in \cS} \left( \weight{\pa(s)}{s}{X_{\pa(s)}}^{-1} \times  \frac{g_s(X_s)}{\prod_{u \in \pa(s)} g_{u\pf s}(X_u)}\right).  \]
By Equation \eqref{eq:split_tochilds}, the second term in brackets can be written as 
\[ \frac{ \prod_{s\in\cS}\prod_{t\in \ch(s)} g_{s\pf t}(X_s) }{ \prod_{s\in\cS}\prod_{u \in \pa(s)} g_{u\pf s}(X_u)} = \frac{\prod_{v\in \cV} g_{\pa(v)\pf v}(X_{\pa(v)}) }{g_0(x_0)} \] which follows by cancellation of terms (note that the numerator is a product over all edges except those originating from the root node, whereas the denominator is a product over all edges except those ending in a leaf node). Hence 
\begin{equation}\label{eq:Pcirc-P}  \frac{\dd \P^\circ}{\dd \P}(X_\cS) = \frac{\prod_{v\in \cV} g_{\pa(v)\pf v}(X_{\pa(v)}) }{g_0(x_0)} \cdot \left(\prod_{s\in \cS}  \weight{\pa(s)}{s}{X_{\pa(s)}}\right)^{-1} . \end{equation} 
The result now follows from 
\[ \frac{\dd \P^\star}{\dd \P^\circ}(X_\cS) =  \frac{\dd \P^\star}{\dd \P}(X_\cS) \cdot \left(  \frac{\dd \P^\circ}{\dd \P}(X_\cS)\right)^{-1} \]
and substituting \eqref{eq:evidence_}, \eqref{eq:hleaf} and \eqref{eq:Pcirc-P}. 
\end{proof}
The definitions leaves open the choice of $g_{s\pf t}$. 
Recall that for Doob's $h$-transform, $h$ is a marginal likelihood.
It stands to reason that  good choices for $g_{s \pf t}$ are approximations of the  marginal likelihood of the observations $x_{\cV_{\pa(t)}}$, given $x_s$ (and integrating out $x_u, u \in \pa(t)\setminus \{s\}$ using (prior-) information available).

\subsection{Continuous-time guided processes}\label{sec:continuoustime}

Up to this point, we have assumed edges to represent ``discrete time'' Markov-transitions. In some applications it is natural to interpret the transition over an edge $(s,t)$ as the result of evolving  a continuous-time Markov process over a fixed time interval $[S, T]$. 
Let $X\equiv (X_u,, u \in [S, T])$ be a time-homogeneous  $E$-valued Markov process  with full (extended) generator ${\cL}$ with domain ${\cD}({\cL})$. Let $\{{\fF}_t\}$ be a right-continuous history filtration and let $\PP_u$ denote the restriction of $\PP$ to ${\fF}_u$. The process $X$ has a family of evolution kernels $\kappa_{u,v}$ with transition densities $p$. For $S<u<v<T$, 
 \[
 \PP(X_v\in \dd y \mid X_u = x) = \kappa_{u,v}(x, \dd y) =  p(u, x; t,v) \dd y
 \]

The discrete pullback of a function $h_T$ (from vertex $t$, defined on $\bB(E)$) by $\kappa = \kappa_{S,T}$ extends on the interval $u \in (S, T]$ to a continuous $h$-transform by setting
\begin{equation}\label{extendedh}
h(u, x) = (\kappa_{u, T} h_T)(x) = \E[\hvertex{T}{X_T} \mid  X_u = x], \qquad  u \in (S,T].
\end{equation}
For given $h_T$, the $h$-function at vertex $T$,   $h$  satisfies the Kolmogorov backward equation
\begin{equation}\label{eq:kolmogorov}
\begin{split}
& \cA h(u, x) = 0, \qquad u\in (S,T]\\  & \text{subject to } h(T,\cdot) = h_T(\cdot),
\end{split}
\end{equation}
where $\cA = \frac{\partial }{\partial t}  + \cL$ is the generator of the space-time process $(u,X_u)$.\footnote{The space-time process $(t,X_t)$ has full   generator ${\cA}={\cL} + \frac{\partial}{\partial t}$
extending the graph
$\{ (f h, \cA (f h)), f \in \cD(\cL), h \in C[0,T]) \}$, 
 see Theorem 7.1 on  page 221 in  \cite{EthierKurtz1986}. } 
Only for few Markov processes it is possible to solve this equation in closed form leading to a tractable expression for  $h$. However, just as in the discrete-time case, we can specify a function $g$ (not necessarily satisfying \eqref{eq:kolmogorov}) and define the guided process by a change of measure.  This results in a generalisation of the definition of guided proposals for  diffusion processes from \cite{schauer2017guided}.

\begin{rem}
The pullback that we defined for the discrete case can also be seen to satisfy \eqref{eq:kolmogorov}.  To see this, consider the space-(discrete)time process $(t, X_t)$. Consider $t\in \{S,T\}$, where $T$ is a child node of $S$ (in a discrete time Markov chain one may think of $S=t$ and $T=t+1$, with $t$ indexing time). Then we have the following discrete analogue for the operator $\cA$:
\begin{align*}
 (\cA h)(S,x) :&= \EE [ h(T, X_{T}) - h(S,X_S) \mid X_S =x] \\ &=\int h(T,y) \kappa_{S\pf T}(x,\dd y) -h(S,x).
\end{align*}
Therefore, solving $ (\cA h)(S,x)=0$ is equivalent to computing the pullback of $h(T,x)$ through $\kappa_{S\pf T}$. Notationally, in the discrete case we have written $h_T(x)$ rather than $h(T,x)$. 
\end{rem}

We assume the process $(X_t)$ is defined on  $[S,T]$ and consider a fixed initial value $x_S$. 

 \begin{defn}
For $u \in (S,T]$ let 
\[ D_u^g = \frac{g(u,X_u)}{g(S+,x_{S+})} \exp\left(-\int_S^u \frac{{\cA} g}{g} (s, X_s) \dd s \right). \]
Assume $g\in {\cD}({\cA})$ is a positive function such that $D_u^g$ is a ${\fF}_u$-martingale. 
Define the change of measure 
\begin{equation}\label{eq:LRcirc_continuous} \dd \PP_u^\circ = D_u^g \dd \PP_u. \end{equation}
The process $X$, with $X_S=x_S$, under the law $\PP_u^\circ$ is denoted by $X^\circ$ and referred to as the \emph{guided process induced by $h$}.
\end{defn}
As in \cite{PalmowskiRolski2002},  we will  call a function such that $D_u^g$ is a $\fF$-martingale  a {\it good} function (sufficient conditions are given in Proposition 3.2 in \cite{PalmowskiRolski2002}). By formula (1.2) in \cite{PalmowskiRolski2002} the full generator of $X^\circ$, characterising the dynamics of $X^\circ$ has the form
\begin{equation}\label{eq:Lcirc} {\cA}^\circ f = \frac1{g} \left( {\cA}(f g) - f {\cA}  g\right).  \end{equation}
This implies $g{\cL}^\circ f =  {\cL}(f g) - f {\cL}  g$ which  will be used in multiple examples to identify/recognise the dynamics of the guided process. Note that the change of measure in this definition is similar to the change of measure in Definition \ref{defn: guided process}.

If $h$  satisfies \eqref{eq:kolmogorov}, then we define $\dd \PP_u^\star = D_u^h \dd \PP_u$ and denote the process under the law of $\PP^\star$ as $X^\star$ having full generator  $\frac1{h} \left( {\cA}(f h) - f {\cA}  h\right)$.

We obtain the following extension of Theorem \ref{thm:lr_xstar_xcirc_dag}.

\begin{thm}\label{thm:continuous_time_weight}
Consider the guided process on a DAG as defined in Definition \ref{defn: guided process}. Suppose that $t\in \cS$ and assume that the transition kernel on the edge $(s,t)$  corresponds to evolving a continuous-time guided process $X^\circ$ on $[S, T]$ with full (extended) generator $\cL$. Then the weight \eqref{eq:dagweights} over the edge equals
\[ w_{s\pf t}\left(X^\circ\right) =  \exp\left(\int_S^T \frac{{\cA} g}{g} (u, X^\circ_u) \dd u \right).\] 
\end{thm}
\begin{proof}
This follows from the proof of Theorem \ref{thm:lr_xstar_xcirc_dag}, from which it is seen that the factor $g(T, X_T^\circ)/g(S, X_S^\circ)$ cancels with factors originating from  parent and child branches.
\end{proof}
One way to choose a tractable guided process  is to consider a second continuous-time Markov process, with generator $\tilde\cL$, and take $g$  solving 
\begin{equation}\label{eq:kolm_tilde}
	(\tilde\cL +\tfrac{\partial}{\partial u}) g(u,x)=0,\quad \text{subject to}\quad  g(T,\cdot)=g_T(\cdot).
\end{equation} 
In this case
$\cA g = (\cL -\tilde \cL) g$, which leads to the following corollary. 

\begin{cor}\label{cor: markov guided}
Assume $h$ and $g$ satisfy \eqref{eq:kolmogorov} and \eqref{eq:kolm_tilde} respectively. 
Suppose both $h$ and $g$ are good functions. Then 
\begin{equation}\label{guided markov}  \frac{\dd \PP_T^\star}{\dd \PP_T^\circ}(X^\circ) = \frac{h_T(X^\circ_T)}{g_T(X^\circ_T)}   
\frac{g(S, X^\circ_S)}{h(S, X^\circ_S)}
\exp\left(\int_S^T \frac{(\cL -\tilde \cL)  g}{g} (s, X^\circ_s) \dd s \right). \end{equation}
\end{cor} 
\begin{rem}
The result is indeed a continuous-time analogue of Theorem \ref{thm:lr_xstar_xcirc_dag}. With \eqref{eq:kolmogorov} the continuous analogue of \eqref{eq:hparent}, $h$ describes the true conditional evolution along the continuous time edge given information encoded in $ h_T$.
It can be used for sampling a continuous-time Markov process conditioned on {\it exact} observations. In this case  $g(t,X^\circ_t)$ and $h(t,X^\circ_t)$ approach the same Dirac measure for $t \to T$. 
In the setting where $X$ is defined by a stochastic differential equation sufficient conditions such that  
 the first ratio on the-right hand-side of equation \eqref{guided markov}  cancels in the limit $t \to T$
have been derived in \cite{schauer2017guided} and \cite{bierkens2020}. In \cite{Corstanje2022}  sufficient conditions are derived for a wider class of Markov processes. 
\end{rem}

Using a different $h$ on the unconditioned dynamics $\kappa$, rather than directly approximating the dynamics of the conditioned process $\kappa^\star$, we aim to approximate the information brought by future observations, and let this approximation guide the process in a natural way. In absence of information from observations, the process evolves just as the unconditional one. This is emphatically the right thing to do to preserve absolute continuity when the graph becomes large.

\section{Compositionality: first steps}\label{sec:compositionality}

In this section we give a principle of compositionality for the task of sampling a guided process using the Backwards Filtering Forward Guiding algorithm.

\subsection{Sequential composition}

To shift to a compositional perspective, it is convenient to frame  the graphical structure, that is the evolution of the process along the edges of a DAG,
differently as \emph{sequential} composition of Markov kernels (as defined in \eqref{eq:chapman}).
We sequentially order the vertices, by partitioning the set $\cS$ in \emph{generations} $\Gamma_1, \Gamma_2, \dots, \Gamma_n $ and $\Gamma_0 = \{0\}$, $\Gamma_{n+1} = \cV$, such that all parents of vertices $s \in \Gamma_i$ are in $\bigcup_{0 \le j < i} \Gamma_j$. 
Write $X_i = X_{\Gamma_i}$ and set
\[
\kappa_i(x_{i-1}; \dd x_{i}) = \P(X_i \in \dd x_i \mid X_{i-1} = x_{i-1}).
\]
We assume $\kappa_{n+1}(x,\cdot)$ has density $p_{n+1}$ with respect to the dominating measure $\lambda$ (consistent with \eqref{eq:density_leaves}). 

Consequently the distribution of $X_i$ is given by a composition of kernels 
\begin{equation}\label{eq:forward decomposition} \delta_{x_0} \kappa_1 \kappa_2\cdots \kappa_i,
\end{equation} obtained by repeatedly pushing forward the initial law (a Dirac measure in $x_0$ as we assumed deterministic $X_0 = x_0$). Parentheses can be omitted by associativity. 
Note that $\{0, \dots, n+1\}$ are the vertices of a line graph with a single observation vertex $n+1$, and Doob's $h$-transform applies on the level of generations.
By repeatedly applying  the pullback operation, we  obtain the \emph{backward filter} 
\begin{equation}\label{eq: hn2}
\begin{split}
h_n(x) &= p_{n+1}(x, x_{n+1})\\
h_{i-1} &= \kappa_i \kappa_{i+1}\cdots \kappa_{n} h_{n}.
\end{split}
\end{equation}

The  $h$-transform maps   $(\kappa_i, h_i)$ to $\kappa_i^\star$, where for  $(\kappa_i, h_i)$, $i   \in \{1, \dots, n\},$
\[
\frac{\kappa_i^\star(x, \dd y)}{\kappa_i(x, \dd y)}  = \frac{h_i(y) }{(\kappa_i h_i)(x)} =: m(x,y)
\]
(not for $\kappa^\star_{n+1}$, which is just a Dirac measure in the observation). 
We will interpret $m$ as a {\it message} obtained from backward filtering, to be used in forward sampling.

\subsection{Forward and backward maps for guided processes}
The following definition is motivated by the backward filter in \eqref{eq: hn2}. In the upcoming definitions we assume kernels $\kappa\colon S \rightarrowtriangle S'$ and $\tilde\kappa\colon S \rightarrowtriangle S'$, where 
$S = (E, \fB)$ and $S' = (E', \fB')$

\begin{defn}\label{def:backwardmap} For a Markov kernel $\kappa\colon S \rightarrowtriangle S'$ and function $h \in \cB(S')$ define the {\it backward map} $\backw\kappa \colon \cB(S') \to \cB(S \times S')\times \cB(S)$ by 
\begin{equation}\label{eq:backw}
\backw{\kappa}(h) = \left(m, \kappa h\right),\quad \text{where} \quad  m(x,y) = \frac{h(y) }{(\kappa h)(x)}.
\end{equation}
\end{defn}
This map returns both the  pullback $\kappa h$  and an appropriate {\it message} $m$ for the map $\forw\kappa$ specified in the following definition.
\begin{defn}\label{def:forwardmap}
For a Markov kernel $\kappa\colon S \rightarrowtriangle S'$ , message $m\in \cB(S \times S')$ (as defined in \eqref{eq:backw}) and measure $\mu \in \cM(S)$ define the {\it forward map} $\forw\kappa: \cB(S \times S')\times \cM(S)\to \cM(S')$ by 
\begin{equation}\label{eq:forw}
 \forw{\kappa}(m, \mu)  = \nu, \quad \nu(\!\dd y) = \int m(x,y) \mu(\!\dd x) \kappa(x, \dd y). 
 \end{equation}
\end{defn}
In case the pullback is intractable or computationally demanding, we propose to replace the kernels $\kappa$ in the backward map by simpler kernels $\tilde\kappa$. Specification of $\tilde\kappa$ is  a way to define $g_{s\pf t}$ in Definition \ref{defn: guided process}.  
If $\mu$ is a probability measure and $\backw\kappa$ sends the message $m$, then $\forw\kappa(m,\mu)$ is again a probability measure. If instead, as for guided processes, $\backw{\tilde\kappa}$ sends the message $m$ (rather than $\backw\kappa$), then $\forw{\kappa}(m,\mu)$ need not be a probability measure, even if $\mu$ is. This motivates the following definition, which is actually an application and suggestive rewriting of the backward- and forward maps.

For a {\it guided process} with backward kernel $\tilde\kappa \colon S \rightarrowtriangle S'$ we have for $g \in \cB(S')$
\[
\backw{\tilde\kappa}(g) = \left(m, \tilde\kappa g\right),\quad \text{where} \quad  m(x,y) = \frac{g(y) }{(\tilde\kappa g)(x)}.
\]
If $\varpi\ge 0$ and $\mu$ is a probability measure then
\[
 \forw{\kappa}(m, \varpi \cdot \mu)(\!\dd y) = (\varpi w_\kappa(m, \mu)) \cdot \nu(\!\dd y) \]
where the {\it weight} $w_\kappa(m, \mu)$ and probability measure $\nu$ are defined by 
\begin{equation}\label{eq:wnu}
\begin{split}
 w_\kappa(m,\mu) &= \iint  m(x, y) \kappa(x,\dd y) \mu(\dd x) =  \int \frac{(\kappa g)(x)}{(\tilde \kappa g)(x)} \mu(\dd x)\\
\nu(\!\dd y) &=  w^{-1}_\kappa(m, \mu) \int  \frac{g(y) }{(\tilde \kappa g)(x)  }  \mu(\!\dd x) \kappa(x, \dd y). 
 \end{split}
 \end{equation}

Clearly, $\forw{\kappa}$ and $\backw{\tilde\kappa}$ work as a pair which we denote  as ${\langle \forw\kappa \mid \backw{\tilde\kappa} \rangle}$ and refer to as {\it optic}. We define the map $F$ that takes kernels $\kappa,\tilde\kappa\colon S \rightarrowtriangle S'$ and returns the optic
\begin{equation}\label{eq:functor}
F(\kappa, \tilde\kappa) = \optic{\kappa}{\tilde\kappa}.\end{equation}

\subsection{Backward Filtering Forward Guiding}
With these definitions, a first form of compositionality is easily obtained for the forward evolution on the level of generations (cf.\ Equation \eqref{eq:forward decomposition}). Hence, assume the forward evolution $\delta_{x_0} \kappa_1 \kappa_2\cdots \kappa_n$. Then sampling from the guided process is obtained by iteratively applying the forward and backward maps in the following way: 
\begin{equation}\label{eq:fb_linegraph}
\begin{split}
m_i, g_{i-1} =\backw{\tilde\kappa_i}(g_i),& \qquad   i=n,\ldots, 1 \qquad 
 \text{and} \\
\omega_i \cdot \delta_{x_i^\circ} \sim  \forw{\kappa_i}(m_i, \omega_{i-1} \delta_{x^\circ_{i-1}}) & \qquad  i=1,\ldots, n,
\end{split}
\end{equation}
initialised from $h_n$, $\omega_0=1$ and $x_0^\circ=x_0$. Here, for an unnormalised bounded measure $\nu$, we write
$\varpi \cdot \delta_{x} \sim \nu$ if $x \sim \nu/\int \nu $ and $\varpi = \int \nu$.

 Hence, for each $i \in \{1,\ldots, n\}$, $x_i^\circ$ is a sample drawn   independently from $\forw{\kappa_i}(m_i, \omega_{i-1} \delta_{x^\circ_{i-1}})/\omega_{i}$. 
Then $\omega_n \cdot (x^\circ_0,\ldots, x^\circ_n)$ is a weighted sample from the conditioned process, i.e.
\begin{equation}\label{eq:weighted_sample} \E [f(X_0, X_1,\ldots, X_n) \mid X_\cV=x_\cV] = \E [\omega_n f(X_0^\circ, X_1^\circ,\ldots, X_n^\circ)]. \end{equation}
 
As we backward filter the process using the kernel $\tilde\kappa$, followed by forward propagating a weighted sample by guiding we call this procedure  {\it backward filtering forward guiding}.

We introduce a {\it message passing diagram} to display the compositional structure of \eqref{eq:fb_linegraph}.
\begin{equation}\label{optic}
\begin{tikzpicture}
    \node[vert] (forw) at (0, 0) {$\forw{}$};
    \node[vert] (backw) at (4, 0) {$\backw{}$};

    \node (muout) [above left =0.8 of  forw] {$\omega' \cdot \delta_{(x^\circ)'}$};
    \node (mu) [below right = 0.8 of forw] {$\omega\cdot \delta_{x^\circ}$};
    \node (h) [above right =0.8 of  backw] {$g'$};
    \node (backwh) [below left = 0.8   of backw] {$g$};
	\node (m) [above right = 0.8 and 1.3 of forw] {$m$};

    \draw[<-] (muout) to[out=-30, in=150] (forw);
    \draw[<-] (forw) to[out=-30, in=150] (mu) ;

    \draw[->] (h) to[out=210,in=30] (backw);
    \draw[->] (backw) to[out=210,in=30] (backwh);
     \node[draw,dashed,fit=(mu) (backwh), inner xsep = 8pt] (box) {};
     
  \draw[<-] (forw) to[out=north east, in=west] ++(1,1)
     to ++(2,0)
     to[out=east, in=north west] (backw)
    ;

\end{tikzpicture}\\[0.2in]
\end{equation}
The message flow of the diagram is right to left in the direction of the arrows and illustrates how the output of $\backw{}$ (the message $m$ and pullback $g$) serves as input for $\forw{}$.

Figure \ref{fig:zipper},  which reminds of a physical zipper is a composition of the elementary diagram \eqref{optic},
in case the forward evolution is $\delta_{x_0}\kappa_1\kappa_2\cdots \kappa_4$. 
The diagram should be read starting from the very right, where the observation from leaf $4$ comes in via $g_{3}$. Then the $h$-transform is computed back to the root by moving south-west towards $g_0$. Subsequently, the conditional marginal is propagated from $\mu_0$ onwards by moving in north-west direction. Simulation of the conditioned process follows the same pattern. The diagram reveals the dependency structure and recursive nature of the algorithm for computing the $h$-transform.

A corresponding pseudo-program in unrolled form, consisting of a function accepting $g_3$ and $x_0$ and sampling, has a similar structure:
\begin{snippet}
function transform123(x0, g3)
    message3, g2 = backward(kernel3, g3)
    message2, g1 = backward(kernel2, g2)
    message1, g0 = backward(kernel1, g1)

    omega1, x1 = forwardsample(kernel1, message1, omega0, x0)
    omega2, x2 = forwardsample(kernel2, message2, omega1, x1)
    omega3, x3 = forwardsample(kernel3, message3, omega2, x2)

    return h0, omega3, (x0,x1,x2,x3) # return evidence and weighted sample
end
\end{snippet}

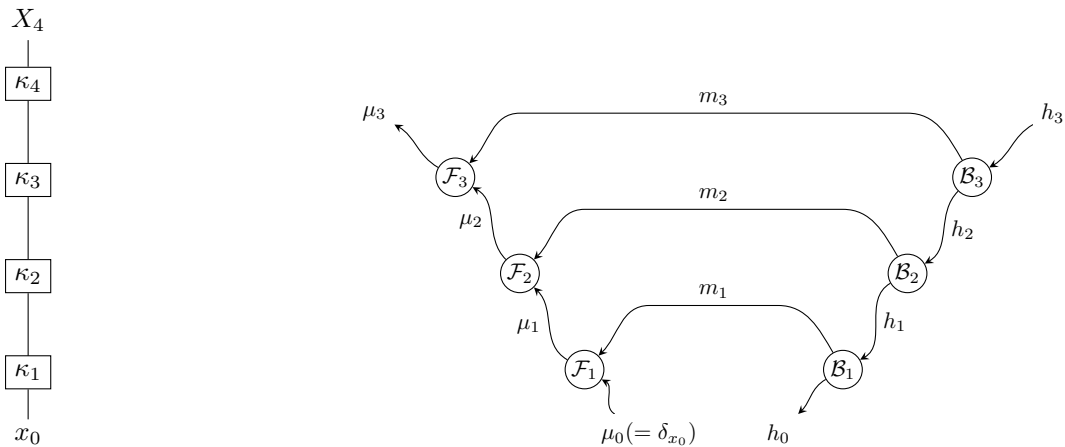
\begin{figure}
\begin{center}
\begin{tikzpicture}[style={scale=0.85}]
	\begin{pgfonlayer}{nodelayer}
		\node [] (1) at (6, -4) {$x_0$};
		
		\node [style=morphism] (23) at (6, -3) {$\kappa_{1}$};
		\node [style=morphism] (24) at (6, 0) {$\kappa_3$};
		\node [style=morphism] (32) at (6, -1.5) {$\kappa_{2}$};
		\node [style=morphism] (36) at (6, 1.5) {$\kappa_4$};
	        \node [] (2) at (6, 2.5) {$X_4$};
	
	\end{pgfonlayer}
	\begin{pgfonlayer}{edgelayer}
			\draw (1) to (23);
		\draw (32) to (24);
		\draw (32) to (23.center);
		\draw (24) to (36);
		\draw (36) to (2);
		
	\end{pgfonlayer}
\end{tikzpicture}\hfill\;
\begin{tikzpicture}[scale=0.85, every node/.style={scale=0.85}]
	\begin{pgfonlayer}{nodelayer}
		\node [style=vert] (0) at (-2, 2.5) {$\forw1$};
		\node [style=vert] (1) at (-3, 4) {$\forw{2}$};
		\node [style=vert] (2) at (-4, 5.5) {$\forw3$};

		\node [style=vertblank] (5) at (-5.25, 6.5) {$\omega_3\cdot \delta_{x^\circ_3}$};
		\node [style=vert] (6) at (2, 2.5) {$\backw1$};
		\node [style=vert] (7) at (3, 4) {$\backw2$};
		\node [style=vert] (8) at (4, 5.5) {$\backw3$};
		\node [style=vertblank] (10) at (5.25, 6.5) {$g_{3}$};
		\node [style=vertblank] (11) at (-1, 1.5) {$1\cdot \delta_{x_0}$};
		\node [style=vertblank] (12) at (1, 1.5) {$g_0$};
	
		\node [style=none] (last) at (6.2, 7) {};
		\node [style=none] (first) at (-6.2, 7) {};
		
	\end{pgfonlayer}
	\begin{pgfonlayer}{edgelayer}
		\draw[->]  [in=30, out=-150] (10) to  node[midway,right] {}  (8);
		\draw[->]  [in=30, out=-135] (8) to node[midway,right] {$g_{2} $}  (7);
		\draw[->]  [in=30, out=-150] (7) to node[midway,right] {$g_{1}$} (6);
		\draw[->]  [in=45, out=-150] (6) to  (12);
		\draw[->]  [in=-30, out=150] (11) to  (0);
		\draw[->]  [in=-45, out=150] (0) to node[midway,left] {$\omega_1\cdot \delta_{x^\circ_1}$}  (1);
		\draw[->]  [in=-30, out=135] (1) to  node[midway,left] {$\omega_2\cdot \delta_{x^\circ_2}$} (2);
		\draw[->]  [in=-30, out=150] (2) to  node[midway,right] {} (5)  ;
		
		\draw[->]    (6)   to[out=120,in=0] ++(-1,+1)  -- ($(0) + (1,+1)$) to[out=180,in=45]  (0)  node[midway,above] {$m_{1}$};
		\draw[->]   (7)   to[out=120,in=0] ++(-1,+1)  -- ($(1) + (1,+1)$) to[out=180,in=45]  (1)  node[midway,above] {$m_{2}$};
		\draw[->]   (8)   to[out=120,in=0] ++(-1,+1) -- ($(2) + (1,+1)$) to[out=180,in=45] (2)  node[midway,above] {$m_{3}$};
	\end{pgfonlayer}
	
\end{tikzpicture}
\end{center}
\caption{Left: String diagram of Markov kernels for a Markov process $x_0$, $X_1$, \dots $X_4$.  Right: Corresponding message passing diagram illustrating passing of arguments between backward operators $\backw i$ and forward operators $\forw{i}$ during backward filtering, forward smoothing for the process observed at the final state $X_4$. 
 Algorithmic time runs from right to left, starting from $g_3$ determined by $x_4$  in the upper right corner. The messages $m_i$ are passed  between backward pass and forward pass from right to left.}
\label{fig:zipper}
\end{figure}

\begin{rem}
The message $m$ from $\backw{\tilde\kappa}$ could equivalently be transmitted as pair $(g, \tilde\kappa g)$ and then combined in the recipient $\forw\kappa$, or even just as $g$, with $\tilde\kappa g$ computed in $\forw\kappa$. We choose  the quotient form as it is convenient for the exposition and acknowledges that while other choices are sometimes convenient, they are essentially equivalent: the format of the message $m$ is something between $\forw\kappa$ and $\backw{\tilde\kappa}$. Also note that the forward map is invariant under multiplying $g$ by an arbitrary strictly positive constant $c$, as the constant is cancelled in the message $m$. 
\end{rem}

\newcommand{\dunderline}[1]{\underline{\underline{#1}}}

\section{Unpacking the compositional structure within generations}\label{sec:unpacking}
In the previous section, backward filtering, forward guiding was only introduced on the level of generations. Just applying the $h$-transform at the level of generations is unsatisfactory, as the additional structure that the Markov kernels $\kappa_i$ inherit from the graphical model is ignored.  In this section we show how to deal with this using string diagrams, duplication- and product-kernels.

\subsection{String diagrams}
Our use of string diagrams originates from applied category theory. See e.g.\ Chapter 3.1 in \cite{jacobs2019structured}, \cite{Cho_2019} or \cite{Selinger2011}. These diagrams give a representation of a probabilistic model on a DAG by explicitly representing Markov kernels which ``sit'' on the edges of a graphical model. 
Typically, the joint law of $(X_s, s \in \Gamma_i)$ is not tractable, all one has  is a mechanism of creating samples $X_i = X_{\Gamma_i}$ using the forward model.
It can be described 
 using products of kernels with the identity kernel $\id$ as neutral element for the product and the special duplication Markov kernel $\Dup$,  which we now define. 
\begin{defn}
Define the ($k$-fold) {\it duplication  kernel} as the Markov kernel 
\begin{equation}\label{eq:def_Dup}
\Dup_k(x, \dd y) = \delta_{x}(\!\dd y_1) \dots \delta_x(\!\dd y_k)	
\end{equation}
and let $\Dup =\Dup_2$. 
\end{defn}

The duplication kernel allows two different kernels $\kappa_1, \kappa_2$ to apply
  independently to the same argument, giving the probability kernel  $(\Dup \comp (\kappa_1 \otimes \kappa_2))(x,
  \cdot) = (\kappa_1 \otimes \kappa_2)((x, x), \cdot)$ representing a conditionally independent transition from one parent to two children (as part of a DAG).   
    
Now we are able to express in detail the conditional distribution $\kappa_i$ of a generation of children $\Gamma_i$ given their parents (which may belong to different generations $0 \le j < i$. 
Each $\kappa_i$ can be written as product of $\kappa_{p \to s}$ (where $p\in \pa(s)$),  duplication kernels and identity kernels $\id$  acting on $X_{i-1}$.    
\begin{figure}
\begin{center}
\hfill{}
\begin{tikzpicture}[style={scale=0.8}]
	\tikzstyle{empty}=[fill=white, draw=black, shape=circle,inner sep=1pt, line width=0.7pt]
	\tikzstyle{solid}=[fill=black, draw=black, shape=circle,inner sep=1pt,line width=0.7pt]
	\begin{pgfonlayer}{nodelayer}
		\node [style=solid,label={$1$},] (0) at (-3.75, -1) {};
		\node [style=solid,label={$2a$}] (1) at (-2, 0) {};
		\node [style=empty,label={$2b$},] (4) at (-2, -2) {};
		\node [style=empty,label={$3a$}] (2) at (-0.25, 1) {};
		\node [style=empty,label={$3b$}] (3) at (-0.25, -1) {};
	\end{pgfonlayer}
	\begin{pgfonlayer}{edgelayer}
		\draw [style=edge] (0) to (1);
		\draw [style=edge] (0) to (4);
		\draw [style=edge] (1) to (2);
		\draw [style=edge] (1) to (3);
	\end{pgfonlayer}
\end{tikzpicture}
\hfill{}\begin{tikzpicture}[style={scale=0.5}]
	\begin{pgfonlayer}{nodelayer}
		\node [style=morphism] (23) at (7, -4.5) {$\qquad\quad\kappa_{1}\quad\qquad$};
		\node [style=bn] (24) at (6, 0) {};
		\node [style=bn] (35) at (7, -3) {};
		
		\node [style=none] (25) at (5, 1) {};
		\node [style=none] (26) at (7, 1) {};
		\node [style=none] (27) at (7, 1) {};
		
		\node [style=morphism] (28) at (5, 1.5) {$\kappa_{3a}$};
		\node [style=morphism] (29) at (7, 1.5) {$\kappa_{3b}$};
		\node [style=morphism] (33) at (9, 1.5) {$\kappa_{2b}$};
		
		\node [style=none] (30) at (5, 3) {};
		\node [style=none] (31) at (7, 3) {};
		\node [style=none] (34) at (9, 3) {};
		
		\node [style=none] (36) at (9, -1) {};
		
		\node [style=morphism] (32) at (6, -1.5) {$\kappa_{2a}$};
	\end{pgfonlayer}
	\begin{pgfonlayer}{edgelayer}
		\draw [style=none, bend left=45] (24) to (25.center);
		\draw [style=none, bend right=45] (24) to (26.center);
		\draw (28) to (25.center);
		\draw (29) to (27.center);
		\draw (30.center) to (28);
		\draw (31.center) to (29);
		\draw (34.center) to (33);
		
		\draw (32) to (24);
		\draw (23) to (35);
		\draw [style=none, bend left=45] (35) to (32.center);
		\draw [style=none, bend right=45] (35) to (36.center);
		\draw (36.center) to (33);
	\end{pgfonlayer}
\end{tikzpicture}\hfill\;\\
\end{center}
\caption{ A simple DAG with three leaves (left) and the corresponding string diagram (right) with 3 observables reaching the upper figure boundary. The root $0$ is not drawn. The black bullet symbol  denotes the duplication (copy) $\Dup$. The string diagram can be written as $\kappa_1 \Dup ( \kappa_{2a}\otimes \id ) (\Dup\otimes\id ) (\kappa_{3a}\otimes \kappa_{3b} \otimes \kappa_{2b})$.}
\label{fig:stringtodag}
\end{figure}

Figure \ref{fig:stringtodag} illustrates the interplay between $\Dup$ and $\otimes$.
The string diagram on the right corresponds to the composition
\begin{equation}\label{eq:comp}
\kappa_1 \; \Dup ( \kappa_{2a}\otimes \kappa_{2b} ) \; (\Dup\otimes\id ) (\kappa_{3a}\otimes \kappa_{3b} \otimes \id).
\end{equation}
The black bullet symbol in Figure \ref{fig:stringtodag} denotes the duplication (copy) $\Dup$. 
As $\kappa \id = \kappa $ and $\kappa \otimes \id = \id \otimes \kappa$, the composition in \eqref{eq:comp} is not unique and an equivalent decomposition is given by 
\[
\kappa_1 \Dup ( \kappa_{2a}\otimes \id ) (\Dup\otimes\id ) (\kappa_{3a}\otimes \kappa_{3b} \otimes \kappa_{2b}),
\]
As we require that the kernels at the leaves have a $\lambda$-density in equation \eqref{eq: hn2} it is this decomposition that we intend to use.   
\begin{ex}
For a state-space model both the ``familiar'' graphical model, as also the corresponding string-diagram are depicted in   Figure \ref{fig:ssm}.
\end{ex}
\begin{figure}
\begin{center}
\begin{tikzpicture}[style={scale=0.8}]
	\tikzstyle{empty}=[fill=white, draw=black, shape=circle,inner sep=1pt, line width=0.7pt]
	\tikzstyle{solid}=[fill=black, draw=black, shape=circle,inner sep=1pt,line width=0.7pt]
	\begin{pgfonlayer}{nodelayer}
		\node [style=solid,label=below:{$0$},] (00) at (-6, 0) {};
		\node [style=solid,label=below:{$1$},] (0) at (-4, 0) {};
		\node [style=empty,label={$y_1$},] (1obs) at (-4, 1.5) {};
		\node [style=solid,label=below:{$2$}] (1) at (-2, 0) {};
				\node [style=empty,label={$y_2$},] (2obs) at (-2, 1.5) {};
		\node [style=solid,label=below:{$3$}] (2) at (-0, 0) {};
				\node [style=empty,label={$y_3$},] (3obs) at (-0, 1.5) {};
		\node [style=none] (end) at (1.0, 0) {.};
	\end{pgfonlayer}
	\begin{pgfonlayer}{edgelayer}
		\draw [style=edge] (00) to (0);
		\draw [style=edge] (0) to (1);
		\draw [style=edge] (1) to (2);
		\draw [style=edge] (0) to (1obs);
		\draw [style=edge] (1) to (2obs);
		\draw [style=edge] (2) to (3obs);
		\draw [style=dashed box] (2) to (end);
	\end{pgfonlayer}
\end{tikzpicture}\\[2em]
\begin{tikzpicture}[style={scale=0.64}]
	\begin{pgfonlayer}{nodelayer}
		\node [style=none] (00) at (-3.0, 1) {};
		\node [style=morphism] (0) at (-1.5, 1) {$\kappa_1$};
		\node [style=bn] (1) at (-0.25, 1) {};
		\node [style=none] (3) at (9, 1) {$y_2$};
		\node [style=bn] (4) at (2.25, 0) {};
		\node [style=morphism] (6) at (6, 0) {$\kappa'_3$};
		\node [style=morphism] (7) at (3.5, -1) {$\kappa_3$};
		\node [style=bn] (8) at (4.75, -1) {};
		\node [style=none] (9) at (6, -2) {};
		\node [style=morphism] (10) at (3.5, 1) {$\kappa'_2$};
		\node [style=none] (11) at (9, 2) {$y_1$};
		\node [style=morphism] (12) at (1, 0) {$\kappa_2$};
		\node [style=morphism] (13) at (1, 2) {$\kappa'_1$};
		\node [style=none] (14) at (9, 0) {$y_3$};
		\node [style=none] (15) at (7.25, -2) {};
	\end{pgfonlayer}
	\begin{pgfonlayer}{edgelayer}
		\draw [style=morphism, bend right=45, looseness=1.25] (8) to (9);
		\draw [style=morphism, bend left=45, looseness=1.25] (8) to (6);
		\draw [style=morphism] (7) to (8);
		\draw [style=morphism, bend right=45, looseness=1.25] (4) to (7);
		\draw [style=morphism, bend left=45, looseness=1.25] (4) to (10);
		\draw [style=morphism] (10) to (3);
		\draw [style=morphism, bend right=45, looseness=1.25] (1) to (12);
		\draw [style=morphism, bend left=45, looseness=1.25] (1) to (13);
		\draw [style=morphism] (0) to (1);
		\draw [style=morphism] (13) to (11);
		\draw [style=morphism] (6) to (14);
		\draw [style=morphism] (12) to (4);
		\draw [style=morphism] (00) to (0.center);
		\draw [style=dashed box] (9) to (15.center);

	\end{pgfonlayer}
\end{tikzpicture}\\
\end{center}
\caption{Top: DAG for state-space model. Bottom:  the corresponding string diagram, rotated 90$^\circ$ to the right. Here,  $\kappa_i$ denote the morphism from  node $i-1$ to $i$ and $\kappa'_i$ is the morphism of node $i$ to the observations $x_i$. \label{fig:ssm} }

\end{figure}

\subsection{Forward and backward  maps for the duplication and product kernels}\label{subsec:duplication kernel}

As the duplication kernel is a Markov kernel, forward and backwards maps are well defined. 
\begin{prop}\label{prop:forward_backward_Delta}
For the $k$-fold duplication kernel  $\Dup_k$ defined in \eqref{eq:def_Dup} we have
\[
\backw{\Dup_k}(h) =  \left(m,\Dup_k h\right), \quad m(x,y) = \frac{h(y_1,\ldots, y_k)}{(\Dup_k h)(x)},
\]
where $(\Dup_k h)(x) = h(x,\ldots,x)$ ($k$ times). 
and
\[
\forw{\Dup_k}(m, \mu)  = f_{\Dup_k}(\mu),
\]
with $ f_{\Dup_k}(\mu)$ the image measure of $\mu$ under the diagonal map $ f_{\Dup_k}\colon x \mapsto (x, \dots, x)$ ($k$ elements).
If $\mu$ is a probability measure, then so is $\nu$.
\end{prop}
The proof follows directly from the definition of the forward- and backward maps. 

Note in particular that  sampling $\forw{\Dup_k}(m, \mu)$ means drawing $z$ from $\mu/\int
\! \dd \mu$ and setting $y_1=\cdots = y_k = z$ with weight equal to $(\int \mu \dd x)^{1/k}$ (other choices are possible as long as the product of the weights is the same).
The backward- and forward  maps can be applied to product kernels $\kappa_1\otimes \kappa_2$ upon direction application of their definitions (Definition \ref{def:backwardmap}  and Definition \ref{def:forwardmap} respectively). 
The following results shows that in case $h$ has special structure, this structure is preserved in the backward map for both the duplication and product kernel.
\begin{cor}
For $h \in \bB(S), h' \in \bB(S')$, if we define  the pointwise product as
\begin{equation}\label{eq:product_h}
(h\odot h')(x,y) = h(x)\cdot h'(y),\qquad x\in E,\, y \in E' 
\end{equation} 
then
$\Dup_k (h_1\odot \cdots \odot h_k) = \prod_{i=1}^k h_i$ and 
$(\kappa \otimes \kappa')(h \odot h')   = (\kappa h) \odot (\kappa' h')$
\end{cor}
In this way we recover the $h$-transform presented in Section \ref{subsec: doob}, in particular Equation  \eqref{eq:split_tochilds} (the {\it fusion} operation). For this, note that $h_n$ in \eqref{eq: hn2}, 
  corresponding to initialisation of $h$ from the leaves, has product form, with components as in \eqref{eq:hleaf}.  

As a tree consists of composing products of ``ordinary'' kernels and duplication kernels we have all ingredients for backwards filtering and forward guiding on a tree.
We  show this using the following representative example:

\begin{figure}
\begin{center}
\hfill{}
\begin{tikzpicture}
	\tikzstyle{empty}=[fill=white, draw=black, shape=circle,inner sep=1pt, line width=0.7pt]
	\tikzstyle{solid}=[fill=black, draw=black, shape=circle,inner sep=1pt,line width=0.7pt]
	\begin{pgfonlayer}{nodelayer}
		\node [style=solid,label={$1$}] (0) at (-4.75, 0) {};
		\node [style=solid,label={$2$}] (1) at (-3, 0) {};
		\node [style=solid,label={$4a$}] (2) at (-0.25, 1) {};
		\node [style=solid,label={$4b$}] (3) at (-0.25, -1) {};
		\node [style=empty,label={$5a$}] (4) at (1.75, 1) {};
		\node [style=empty,label={$5b$}] (5) at (1.75, -1) {};
		
	\end{pgfonlayer}
	\begin{pgfonlayer}{edgelayer}
		\draw [style=edge] (0) to (1);
		\draw [style=edge] (1) to (2);
		\draw [style=edge] (1) to (3);
		\draw [style=edge] (2) to (4);
		\draw [style=edge] (3) to (5);
	\end{pgfonlayer}
\end{tikzpicture}
\hfill{}\begin{tikzpicture}[style={scale=0.5}]
	\begin{pgfonlayer}{nodelayer}
	
		\node [style=bn] (24) at (6, 0) {};
		\node [style=none] (25) at (5, 1) {};
		\node [style=none] (26) at (7, 1) {};
		\node [style=none] (27) at (7, 1) {};
		\node [style=morphism] (28) at (5, 1.5) {$\kappa_{4a}$};
		\node [style=morphism] (29) at (7, 1.5) {$\kappa_{4b}$};
		\node [style=morphism] (30) at (5, 3) {$\kappa_{5a}$};
		\node [style=morphism] (31) at (7, 3) {$\kappa_{5b}$};
			\node [style=morphism] (23) at (6, -3)  {$\qquad\kappa_{1}\qquad$};

		\node [style=none] (33) at (5, 4.2) {};
		\node [style=none] (34) at (7, 4.2) {};
		\node [style=morphism] (32) at (6, -1.5) {$\qquad\kappa_{2}\qquad$};
		
	\end{pgfonlayer}
	\begin{pgfonlayer}{edgelayer}
		\draw [style=none, bend left=45] (24) to (25.center);
		\draw [style=none, bend right=45] (24) to (26.center);
		\draw (28) to (25.center);
		\draw (29) to (27.center);
		\draw (30.center) to (28);
		\draw (31.center) to (29);
		\draw (32) to (24);
		\draw (32) to (23.center);	
		\draw (33.center) to (30);
		\draw (34.center) to (31);
	\end{pgfonlayer}
\end{tikzpicture}\hfill\;
\hfill{}\begin{tikzpicture}[style={scale=0.5}]
	\begin{pgfonlayer}{nodelayer}
		\node [style=morphism] (23) at (6, -3)  {$\qquad\kappa_{1}^\star\qquad$};
		\node [style=bn] (24) at (6, 0) {};
		\node [style=none] (25) at (5, 1) {};
		\node [style=none] (26) at (7, 1) {};
		\node [style=none] (27) at (7, 1) {};
		\node [style=morphism] (28) at (5, 1.5) {$\kappa^\star_{4a}$};
		\node [style=morphism] (29) at (7, 1.5) {$\kappa^\star_{4b}$};
		\node [style=none] (30) at (5, 4.2) {};
		\node [style=none] (31) at (7, 4.2) {};
		\node [style=morphism] (32) at (6, -1.5) {$\qquad\kappa^\star_{2}\qquad$};

	\end{pgfonlayer}
	\begin{pgfonlayer}{edgelayer}
		\draw [style=none, bend left=45] (24) to (25.center);
		\draw [style=none, bend right=45] (24) to (26.center);
		\draw (28) to (25.center);
		\draw (29) to (27.center);
		\draw (30.center) to (28);
		\draw (31.center) to (29);
		\draw (32) to (24);
		\draw (32) to (23.center);	
	\end{pgfonlayer}
\end{tikzpicture}\\\end{center}
\caption{A  tree with two leaves and vertices numbered according to level in the corresponding string diagram of the composition of $\kappa_1 \kappa_2 \kappa_3 \kappa_4 \kappa_5 =\kappa_1 \comp \kappa_{2} \comp \Dup\comp (\kappa_{4a}\otimes \kappa_{4b})\comp (\kappa_{5a}\otimes \kappa_{5b})$.  $\kappa^\star_{5a}(x, \cdot) = \delta_{\textstyle x_{5a}}$ and $\kappa^\star_{5b}(x', \cdot) = \delta_{\textstyle x_{ 5b}}$ which are Dirac measures independent of their argument recovering the observations, are not drawn.}
\label{fig:noncollider}
\end{figure}

\begin{ex}[Backward-forward diagram on a directed tree]
Consider the directed tree and string diagram in Figure~\ref{fig:noncollider}. Assume we backward filter with kernels $\tilde\kappa_i$. Abbreviate $\backw{i}$ for $\backw{\tilde\kappa_i}$ and $\forw{i}$ for $\forw{\kappa_i}$. Corresponding to the string-diagram we have the following the  following dependency diagram which shows the compositional structure of forward and backward maps for this example:

\begin{equation}\label{diag:forwardbackward_noncollider}
\begin{tikzpicture}[scale=0.85, every node/.style={scale=0.85}]
	\begin{pgfonlayer}{nodelayer}
		\node [style=vert] (0) at (-2, 2.5) {$\forw1$};
		\node [style=vert] (1) at (-3, 4) {$\forw{2}$};
		\node [style=vert] (2) at (-4, 5.5) {$\forw3$};
		\node [style=vert] (4) at (-4.5, 7) {$\forw{4b}$};
		\node [style=vert] (5) at (-5.5, 7) {$\forw{4a}$};
		\node [style=vert] (6) at (2, 2.5) {$\backw1$};
		\node [style=vert] (7) at (3, 4) {$\backw2$};
		\node [style=vert] (8) at (4, 5.5) {$\backw3$};
		\node [style=vert] (9) at (4.5, 7) {$\backw{4b}$};
		\node [style=vert] (10) at (5.5, 7) {$\backw{4a}$};
		\node [style=vertblank] (11) at (-1, 1.5) {$1 \cdot\delta_{x_0}$};
		\node [style=vertblank] (12) at (1, 1.5) {$g_0$};
	
		\node [style=none] (last) at (6.2, 7) {};
		\node [style=none] (first) at (-6.2, 7) {};
		
		\node [style=vertblank] (out) at (7.5, 7) {$(g_{4b}, g_{4a})$};
		\node [style=vertblank] (in) at (-8.5, 7) {$(\varpi_{4a}\cdot\delta_{x^\circ_{4a}}, \varpi_{4b}\cdot\delta_{x^\circ_{4b}})$};
	\end{pgfonlayer}
	\begin{pgfonlayer}{edgelayer}
		\draw[->]  [in=60, out=-120] (9) to  node[midway,left] {$g_{3b}$}  (8);
		\draw[->]  [in=30, out=-150] (10) to  node[midway,right] {$g_{3a}$}  (8);
		\draw[->]  [in=30, out=-135] (8) to node[midway,right] {$g_{2} (= g_{3a}g_{3b})$}  (7);
		\draw[->]  [in=30, out=-150] (7) to node[midway,right] {$g_{1}$} (6);
		\draw[->]  [in=45, out=-150] (6) to  (12);
		\draw[->]  [in=-30, out=150] (11) to  (0);
		\draw[->]  [in=-45, out=150] (0) to node[midway,left] {$\varpi_1\cdot\delta_{x^\circ_1}$}  (1);
		\draw[->]  [in=-30, out=135] (1) to  node[midway,left] {$\varpi_2\cdot\delta_{x^\circ_2}$} (2);
		\draw[->]  [in=-60, out=120] (2) to  node[midway,right] {$\sqrt{\varpi_{3}}\cdot\delta_{x^\circ_{3}}$}  (4);
		\draw[->]  [in=-30, out=150] (2) to  node[midway,left] {$\sqrt{\varpi_{3}}\cdot\delta_{x^\circ_{3}}$} (5)  ;
		\draw[<-] (last.center) to (out);
		\draw[<-] (in) to (first.center);
		
		\draw[->]    (6)   to[out=120,in=0] ++(-1,+1)  -- ($(0) + (1,+1)$) to[out=180,in=45]  (0)  node[midway,above] {$m_{1}$};
		\draw[->]   (7)   to[out=120,in=0] ++(-1,+1)  -- ($(1) + (1,+1)$) to[out=180,in=45]  (1)  node[midway,above] {$m_{2}$};
		\draw[->,dashed]   (8)   to[out=120,in=0] ++(-1,+1) to ($(2) + (1,+1)$) to[out=180,in=45] (2);
		\draw[->]  (9)  to[out=120,in=0] ++(-1,+1)  -- ($(4) + (1,+1)$) to[out=180,in=45] (4) node[midway,above] {$m_{4b}$};
		\draw[->]   (10)  to[out=120,in=0] ++(-1.7,+1.7)  -- ($(5) + (1.7,+1.7)$) to[out=180,in=45] (5)   node[midway,above] {$m_{4a}$};
	\end{pgfonlayer}
	\node[draw,dashed,fit=(4) (5), inner xsep = 8pt] (box) {};
	\node[draw,dashed,fit=(9) (10), inner xsep = 8pt] (box) {};
	
\end{tikzpicture}
\end{equation}
This diagram should be read starting from the very right, where observations from the leaves $5a, 5b$ come in via $(g_{4a}, g_{4b})$. Then the $h$-transform is computed back to the root by moving south-west towards $g_0$. Subsequently, the conditional marginal is propagated from $\delta_{x_0}$ onwards by moving in north-west direction. Simulation of the conditioned process follows the same pattern, but then we would also include the weights explicitly in the diagram.
\end{ex}

\begin{rem}
In applying the scheme \eqref{eq:forward decomposition} intermediate values $X_i$ are lost. To keep them one can make a duplicate and keep one of these by formally applying the identity kernel $\id$ to it in each subsequent composition.\footnote{A string diagram is called {\it accessible} if all its internal connections are also accessible externally (\cite{jacobs2019structured}, Section 3.3).  
}
By doing so, no longer  all leaves are being observed. To deal with this consider the case where only a subset of the leaves are observed. It suffices to consider $\kappa_{n+1}$ of product form,  $\kappa_{n+1} = \kappa_{(n+1)a} \otimes \kappa_{(n+1)b}$  and assume  only $X_{(n+1)a} = x_{(n+1)a}$ to be observed.
In that caee it suffices to define 
\begin{equation}\label{eq: hn} 
  g_{n}(x)  =   p_{(n+1)a}(x, x_{(n+1)a}).
\end{equation}
\end{rem}

\begin{rem}
So far we have assumed conditionally independent observations. Suppose this is not the case, for example when $X_{n+1} = (Y, Y')$ with $\kappa_{n+1}\colon S_n \rightarrowtriangle T\otimes T'$ and only $Y' = y'$ was observed.  Then
 with $\tilde\kappa\colon S_n \rightarrowtriangle T$, $\tilde\kappa'\colon S_n \rightarrowtriangle T'$, 
 and   distribution of $(\tilde Y_x,\tilde Y'_x)$  given by $\tilde\kappa(x, \cdot)\otimes \tilde\kappa'(x,\cdot)$
\[
\E [f(X_{n+1}) \mid X_n = x, Y'=y'] = \frac{\E\;f(\tilde Y_x, y')\Phi(x, \tilde Y_x, y')}{\E \;\Phi(x, \tilde Y_x, y')},\]
where 
\[ \Phi(x, y, y') =  \frac{\dd \kappa(x, \cdot)}{\dd \tilde\kappa(x, \cdot)\otimes \tilde\kappa'(x,\cdot)}(y,y').
\]
Thus, \eqref{eq: hn2} can be replaced by
$g_{n}(x)  = \E \;\Phi(x, \tilde Y_x, y')$.

\end{rem}


\subsection{Extension from a directed tree to a DAG}

In this section we extend our approach from a directed tree to a true DAG. Whereas on the former each vertex has a unique parent vertex, on a DAG a vertex can have multiple parent vertices. Correspondingly, in the diagram a box with several input edges denotes a Markov kernels  whose domain is a tensor product. 

Note that if  $\kappa = \kappa_1 \comp \kappa_2$,  then $\kappa^\star = \kappa_1^\star \comp \kappa_2^\star$. Moreover, if  $\kappa_2 = \kappa_{2a} \otimes \kappa_{2b}$ has product form, we also get
\[
(\kappa_1\comp  \kappa_{2} )^\star = (\kappa_1\comp  (\kappa_{2a} \otimes \kappa_{2b}))^\star = \kappa_1^\star \comp (\kappa_{2a}^\star \otimes \kappa_{2b}^*).
\]
This corresponds to our earlier observation, that the $h$-transform can be computed in the tree one vertex and its children at a time. It illustrates that the transform as we have considered it until now is  structure preserving on a directed tree. This property is lost in case of a general DAG.

\begin{figure}
\begin{center}
\hfill{}
\begin{tikzpicture}
	\tikzstyle{empty}=[fill=white, draw=black, shape=circle,inner sep=1pt, line width=0.7pt]
	\tikzstyle{solid}=[fill=black, draw=black, shape=circle,inner sep=1pt,line width=0.7pt]
	\begin{pgfonlayer}{nodelayer}
		\node [style=solid,label={$1a$},] (-1) at (-4.75, 1) {};
		\node [style=solid,label={$1b$},] (0) at (-4.75, -1) {};
		\node [style=solid,label={$2$}] (1) at (-3, 0) {};
		\node [style=solid,label={$4a$}] (2) at (-0.25, 1) {};
		\node [style=solid,label={$4b$}] (3) at (-0.25, -1) {};
		\node [style=empty,label={$5a$}] (4) at (1.75, 1) {};
		\node [style=empty,label={$5b$}] (5) at (1.75, -1) {};
		
	\end{pgfonlayer}
	\begin{pgfonlayer}{edgelayer}
		\draw [style=edge] (-1) to (1);
		\draw [style=edge] (0) to (1);
		\draw [style=edge] (1) to (2);
		\draw [style=edge] (1) to (3);
		\draw [style=edge] (2) to (4);
		\draw [style=edge] (3) to (5);
	\end{pgfonlayer}
\end{tikzpicture}
\hfill{}\begin{tikzpicture}[style={scale=0.5}]
	\begin{pgfonlayer}{nodelayer}
	
		\node [style=morphism] (22) at (5, -3) {$\kappa_{1a}$};
		\node [style=morphism] (23) at (7, -3) {$\kappa_{1b}$};
		\node [style=bn] (24) at (6, 0) {};
		\node [style=none] (25) at (5, 1) {};
		\node [style=none] (26) at (7, 1) {};
		\node [style=none] (27) at (7, 1) {};
		\node [style=morphism] (28) at (5, 1.5) {$\kappa_{4a}$};
		\node [style=morphism] (29) at (7, 1.5) {$\kappa_{4b}$};
		\node [style=morphism] (30) at (5, 3) {$\kappa_{5a}$};
		\node [style=morphism] (31) at (7, 3) {$\kappa_{5b}$};

		\node [style=none] (33) at (5, 4.2) {};
		\node [style=none] (34) at (7, 4.2) {};
		\node [style=morphism] (32) at (6, -1.5) {$\qquad\kappa_{2}\qquad$};
		
	\end{pgfonlayer}
	\begin{pgfonlayer}{edgelayer}
		\draw [style=none, bend left=45] (24) to (25.center);
		\draw [style=none, bend right=45] (24) to (26.center);
		\draw (28) to (25.center);
		\draw (29) to (27.center);
		\draw (30) to (28);
		\draw (31) to (29);
		
		\draw (33.center) to (30);
		\draw (34.center) to (31);

		\draw (32) to (24);
		\draw (22.south) to (22.south |- 32.north);
		\draw (23.south) to (23.south |- 32.north);

	\end{pgfonlayer}
\end{tikzpicture}\hfill\;
\hfill{}\begin{tikzpicture}[style={scale=0.5}]
	\begin{pgfonlayer}{nodelayer}
		\node [style=morphism] (23) at (6, -3) {$(\kappa_{1a} \otimes \kappa_{1b})^\star$};
		\node [style=bn] (24) at (6, 0) {};
		\node [style=none] (25) at (5, 1) {};
		\node [style=none] (26) at (7, 1) {};
		\node [style=none] (27) at (7, 1) {};
		\node [style=morphism] (28) at (5, 1.5) {$\kappa^\star_{4a}$};
		\node [style=morphism] (29) at (7, 1.5) {$\kappa^\star_{4b}$};
		\node [style=none] (30) at (5, 4.2) {};
		\node [style=none] (31) at (7, 4.2) {};
		\node [style=morphism] (32) at (6, -1.5) {$\quad\kappa^\star_{2}\quad$};
		
	\end{pgfonlayer}
	\begin{pgfonlayer}{edgelayer}
		\draw [style=none, bend left=45] (24) to (25.center);
		\draw [style=none, bend right=45] (24) to (26.center);
		\draw (28) to (25.center);
		\draw (29) to (27.center);
		\draw (30.center) to (28);
		\draw (31.center) to (29);
		\draw (32) to (24);
		\draw (32) to (23.center);	
	\end{pgfonlayer}
\end{tikzpicture}\\\end{center}
\caption{A  DAG with two leaves and two roots and vertices numbered according to level in the corresponding string diagram of the composition of $\kappa_1 \kappa_2 \kappa_3 \kappa_4 \kappa_5 ={(\kappa_{1a}\otimes \kappa_{1b}) \comp \kappa_{2} \comp \Dup\comp (\kappa_{4a}\otimes \kappa_{4b})\comp (\kappa_{5a}\otimes \kappa_{5b})}$. Conditioning on the end values introduces dependence between roots.
 $\kappa^\star_{5a}(x, \cdot) = \delta_{\textstyle x_{5a}}$ and $\kappa^\star_{5b}(x', \cdot) = \delta_{\textstyle x_{ 5b}}$ which are Dirac measures independent of their argument recovering the observations, are not drawn.}
\label{fig:collider}
\end{figure}
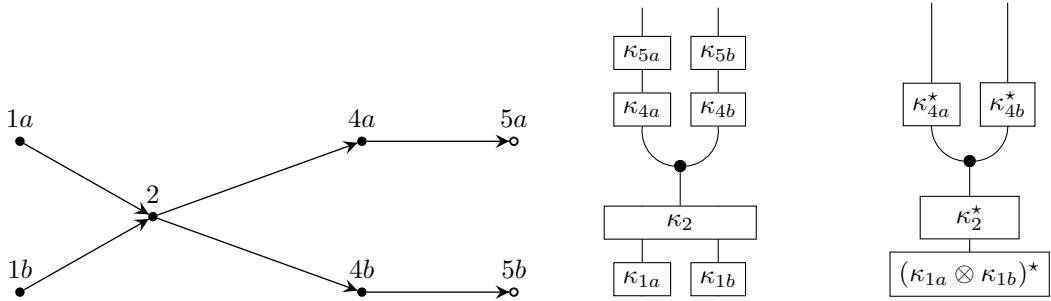

To illustrate the issue, we consider Doob's $h$-transform for  the DAG in Figure \ref{fig:collider}. The rightmost picture shows the situation corresponding to a collider in graphical models. Here  $\kappa_1 = \kappa_{1a}\otimes \kappa_{1b}$ has product form, but  $\kappa_2$ does not. The transformed kernel $\kappa_1^\star$ typically won't have product form, the states become entwined \citep{jacobs2019structured}.

  Consider a kernel with multiple parents, say $\kappa_{\pa(s)\pf s}(x,\dd y)$.    
Its input $x = \{x_u,\, u\in\pa(s)\}$ is a concatenation of values of the process at the parent vertices. For an approximation $\tilde\kappa_{\pa(s)\pf s}(x,\dd y)$ of $\kappa_{\pa(s)\pf s}(x,\dd y)$, we have pullback 
\begin{equation}\label{eq:pullback multipleparents1}
g(x_{\pa(s)}):=(\tilde\kappa_{\pa(s)\pf s} g)(x_{\pa(s)}).\end{equation} To backpropagate this function, we want to have a function in product form $\prod_{u\in \pa(s)} g_u(x_u)$ to each parent individually.  To this end, write $x_{\pa(s)}=(x_1,\ldots, x_d)$ and suppose the density of $x_{\pa(s)}$ is given by $\pi(x_1,\ldots, x_d)$. Then the joint distribution of $X_{\pa(s)}$ and its leaf descendants is  proportional to $\pi(x_1,\ldots, x_d)g(x_1,\ldots, x_d)$. Let $\pi_i$ denote the marginal distribution of $X_i$. Denote $x_{-i}$ the vector $x$ without its $i$-th component. Recall that for $\lambda$-probability densities $p$ and $q$, the Kullback-Leibler divergence of $p$ to $q$ is given by $\on{KL}(p,q) = \int p \log p/q \dd \lambda$. 
\begin{prop}\label{prop: kl_backward}
The minimiser of $\on{KL}(\pi g, \prod_{i=1}^d \pi_i g_i)$ over $g_1,\ldots, g_d$, subject to $\int g_i(x_i) \pi_i(x_i) \dd \lambda(x_i) =c_i$ ($1\le i \le d$) is given by 
\[ g_i(x_i) = c_i^{-1} \EE_\pi [g(X_1,
\ldots, X_d \mid X_i=x_i]. \]
\end{prop}
As a consequence to this proposition (with $c_j=1$ for all $j$) we define $g_u$ by 
\begin{equation}\label{eq:pullback multipleparents2}
		g_u(x_u)   = \int g(x_{\pa(s)}) \pi(x_{-u}\mid x_u) \dd x_{-u},\qquad u\in \pa(s).	
\end{equation}
Upon defining 
\begin{equation}\label{eq:kappa_u_to_s}	 
	\tilde\kappa_{u\pf s}(x_u, \dd y) := \int \tilde\kappa_{\pa(s)\pf s}(x_{\pa(s)}, \dd y) \pi(x_{-u}\mid x_u) \dd x_{-u}, 
\end{equation}
it follows that we can equivalently define $g_u$ in terms of these kernels by
\begin{align*}
g_u(x_u) =& \iint  \tilde\kappa_{\pa(s)\pf s}(x_{\pa(s)}, \dd y) g(y) \dd y\, \pi(x_{-u}\mid x_u) \dd x_{-u} \\
	= & \int \tilde\kappa_{u\pf s}(x_u, \dd y) g(y) \dd y = (\tilde\kappa_{u\pf s}g)(x_u)
\end{align*}
(use Fubini's theorem). 
Here, the  prediction kernels $\tilde\kappa_{u\pf s}(x_u,\cdot)$, $u \in \pa(s)$ serve as \emph{a priori predictor} of $X_s$ if only $x_u$ is known. (From this perspective,  $\kappa_{\pa(s)\pf s}(x,\cdot)$ is the prior predictive distribution of $X_s$ if  $x = \{x_u,\, u\in\pa(s)\}$ is known.)

\begin{defn}\label{defn:pullback_dag}
Assume $\vert\pa(s)\vert\ge 2$. 	
To the kernel $\kappa_{\pa(s)\pf s}$ we define the pullback-kernel  $\bar\kappa_{\pa(s)\pf s}$  	 induced by $\pi(x_{\pa(s)})$ 
  by its action on functions $g \in \bB(S_{s})$ 
  \[(\bar\kappa_{\pa(s)\pf s}g)(x) = \prod_{u\in \pa(s)} (\tilde\kappa_{u\pf s} g)(x_u). \]
 Here, the Markov kernels $\tilde\kappa_{u\pf s}$, $u\in \pa(s)$ are defined in Equation \eqref{eq:kappa_u_to_s}.
\end{defn}

 The operator $\bar\kappa_{\pa(s)\pf s}$ is {\it not} a  Markov kernel operator and therefore we denote it with a bar rather than a tilde.
 From it we obtain  a backward map $\cB_{\bar\kappa_{\pa(s)\pf s}}$ (Definition \ref{def:backwardmap}) 
 \[ \backw{\bar\kappa_{\pa(s)\pf s}}(g) = \left(m, \bar\kappa_{\pa(s)\pf s}g\right), \qquad m(x,y) = \frac{g(y)}{ \prod_{u\in \pa(s)} (\tilde\kappa_{u\pf s} g)(x_u)} \] 
 to be used with the forward map for
 $\kappa_{\pa(s)\pf s}(x,\dd y)\colon S_{\pa(s)} \rightarrowtriangle S$, $S=(E, \fB)$. 
From the form of the message $m$ it follows immediately that the forward map picks up the ``correct'' weight 
according to \eqref{eq:dagweights}. 
Summarising, for a kernel $\kappa(x,\dd y)$ with $x=(x_u,\, u\in \pa(s))$  pullback of $h$ consists of 
\begin{itemize}	
\item an initial ``ordinary'' pullback as in Equation \eqref{eq:pullback multipleparents1};
\item computing $(\bar\kappa_{\pa(s)\pf s} g)(x_{\pa(s)}) = \prod_{u\in \pa(s)} g_u(x_u)$ with $g_u$ as defined in Equation \eqref{eq:pullback multipleparents2}. 
\end{itemize}
Intuitively, whereas the forward evolution would sample $y$ based on its ``complete'' input $x$, the backward kernel uses a collection of kernels, one for  each parent $u\in \pa(s)$.  Each of these kernels samples $y$ based on only $x_u$. 

\begin{rem}\label{rem:changeperspective}
Up to this section our framework has been string diagrams of Markov kernels. 
Note that we make a change of perspective where we focus on the action of kernels on functions $g$, where the action need not be induced by a Markov kernel. 
\end{rem}

\section{Compositionality results for backward filtering forward guiding}\label{sec:comp_bffg}

Introducing a categorical perspective is motivated by the admission that we have already used the concepts introduced shortly in all but name in the previous section   and by the following testament: ``We should approach the problem of statistical modelling and
computation in a modular, composable, functional way, guided by underpinning principles from category theory.'' \cite{talk_Wilkinson2017}.
While we defer a discussion of guided processes from a categorical point of view to a companion paper, we do give a number of compositionality results here. 

\subsection{Composition of kernels}
Recall from \eqref{eq:functor} the definition  $
F(\kappa, \tilde\kappa) = \optic{\kappa}{\tilde\kappa}$. If $\kappa\colon S\rightarrowtriangle S'$ and $\tilde\kappa\colon S\rightarrowtriangle S'$, then $F(\kappa ,\tilde \kappa)$ acts upon a pair consisting  of a bounded function and measure as follows: $F(\kappa ,\tilde \kappa) \colon \cB(S') \times \cM(S)\to \cB(S) \times \cM(S')$ with
\[ F(\kappa ,\tilde \kappa)(g', \mu) = (\tilde\kappa g', \forw{\kappa}(m, \nu)),\] where the message $m$ is defined by $\backw{\tilde\kappa}(g')=(m, \tilde\kappa g')$. 
 Note that if $\kappa=\tilde\kappa=\id$,  then $\kappa g=g$, $m(x,y) = g(y)/g(x)$ and $\forw\kappa(m, \mu) =  \mu$. Therefore, $\langle \forw{\id} \mid \backw{\id} \rangle(g, \mu)=(g, \mu)$. Hence there exists an identity optic $\cI$, given by  $\cI=\optic{\id}{\id}$.

If Markov kernels can be composed, then their induced optics can be composed as well according to the following definition.
\begin{defn}\label{def:composition optics}
For $i\in \{1,2\}$ assume Markov kernels
$\kappa_i  \colon S_{i-1}\rightarrowtriangle S_i$ and $\tilde\kappa_i  \colon S_{i-1}\rightarrowtriangle S_i$.   Suppose 
$\mu \in \cM(S_{0})$. We denote the {\it  composition of the optics} $F(\kappa_1,\tilde\kappa_1)$ with $F(\kappa_2, \tilde\kappa_2)$ by $F(\kappa_1,\tilde\kappa_1)F(\kappa_2,\tilde\kappa_2)$. We define $F(\kappa_1,\tilde\kappa_1)F(\kappa_2,\tilde\kappa_2)  \colon  \cB(S_2) \times \cM(S_0) \to  \cB(S_0) \times \cM(S_2)$  by 
\begin{equation}\label{eq:composition} \left(F(\kappa_1,\tilde\kappa_1)F(\kappa_2,\tilde\kappa_2)\right)(g, \mu) = 
(g_{12}, \mu_{12}),\quad \end{equation}
where $(m_2, g_2) = \backw{\tilde\kappa_2}(g)$,  $(m_1, g_{12})=\backw{\tilde\kappa_1}(g_2)$ and $\mu_{12}=\forw{\kappa_2}(m_2, \forw{\kappa_1}(m_1,\mu))$.  
\end{defn}
While this definition may look somewhat complicated, the following figure facilitates understanding the composition rule:
\begin{equation}
    \begin{tikzpicture}[baseline=(R), every node/.style={scale=0.8}]
        \begin{scope}[on grid]
        \node[vert] (l') at (0, 0) {$\forw2$};
        \node[vert, below right = 0.7 and 1 of l'] (l) {$\forw1$};
        \node[vert] (r') at (5, 0) {$\tilde{\cB}_2$};
        \node[vert, below left = 0.7 and 1 of r'] (r) {$\tilde{\cB}_1$};
        \node (A) [below right = 0.7 and 1 of l] {$\mu$};
        \node (A') [below left = 0.7 and 1 of r] {$g_{12}$};
        \node (R) [left of=l'] {$\mu_{12}$};
        \node (R') [right of=r'] {$g_2$};
        \draw[<-] (R) -- (l');
        \draw[->] (R') -- (r');
        \draw[<-] (l') to[out=north east, in=west] ++(1,1)
         to ++(3,0)
         to[out=east, in=north west] (r')
        ;
        \draw[<-] (l) to[out=north east, in=west] ++(1,1)
         to ++(1,0)
         to[out=east, in=north west] (r)
        ;
        \draw[<-] (l') to[out=south east,in=west] (l);
        \draw[<-] (r) to[out=east, in=south west] (r');
        \draw[<-] (l) to[out=south east,in=west] (A);
        \draw[->] (r) to[out=south west,in=east] (A');
        \node[draw,dashed,fit=(l) (r), inner xsep = 16pt, inner ysep = 30pt] (box2) {};
        \node[draw,dashed,fit=(A) (A'), inner xsep = 8pt] (box) {};
        \end{scope}
        \end{tikzpicture}
\quad
    \begin{tikzpicture}[baseline=(R), every node/.style={scale=0.8}]
        \begin{scope}[on grid]
        \node[vert] (l_1) at (0, 0) {};
        \node[vert, below right = 0.7 and 1 of l'] (l_2) {};
        \node[vert] (r_1) at (5, 0) {};
        \node[vert, below left = 0.7 and 1 of r'] (r_2) {};
        \node (C) [below right = 0.7 and 1 of l] {$\mu$};
        \node (C') [below left = 0.7 and 1 of r] {$g_{12}$};
        \node (R) [left of=l'] {$\mu_{12}$};
        \node (R') [right of=r'] {$g_2$};
        \draw[<-] (R) -- (l');
        \draw[->] (R') -- (r');
        \draw[<-] (l') to[out=north east, in=west] ++(1,1)
         to ++(3,0)
         to[out=east, in=north west] (r')
        ;
        \draw[<-] (l) to[out=north east, in=west] ++(1,1)
         to ++(1,0)
         to[out=east, in=north west] (r)
        ;
        \draw[<-] (l') to[out=south east,in=west] (l);
        \draw[<-] (r) to[out=east, in=south west] (r');
        \draw[<-] (l) to[out=south east,in=west] (A);
        \draw[->] (r) to[out=south west,in=east] (A');
        \node[circle,draw,dashed,fit=(l_1) (l_2), inner sep = 3pt, label=110:$\cF_{1,2}$] (l_12) {};
        \node[circle,draw,dashed,fit=(r_1) (r_2), inner sep = 3pt, label=80:$\tilde{\cB}_{12}$] (r_12) {};
        \node[draw,dashed,fit=(C) (C'), inner xsep = 8pt] (box) {};
        \end{scope}
        \end{tikzpicture}
\end{equation}

The setting of the following theorem is as in the above definition, with $\kappa_{12}=\kappa_1\kappa_2$ and $\tilde\kappa_{12}=\tilde\kappa_1\tilde\kappa_2$
\begin{thm}\label{thm:seqcomp} The optic of the composition of kernels is the same as the composition of their induced optics: \[
F(\kappa_{12}, \tilde\kappa_{12})=F(\kappa_1,\tilde\kappa_1) F(\kappa_2, \tilde\kappa_2)\] 
\end{thm}
\begin{proof}
Assume $\kappa_{12}=\kappa_1\kappa_2$ and $\tilde\kappa_{12}=\tilde\kappa_1\tilde\kappa_2$.
We show that the right-hand-side of \eqref{eq:composition} equals $\optic{\kappa_{12}}{\tilde\kappa_{12}}(g,\mu)$ which is, by definition,
\[ \optic{\kappa_{12}}{\tilde\kappa_{12}}(g,\mu) = (\tilde\kappa_{12} g, \nu) \]
with $\nu(\!\dd y) = \int m_{12}(x,y) \mu(\!\dd x) \kappa_{12}(x, \dd y)$ and $m_{12}(x,y) = g(y)/(\kappa_{12} g)(x)$. Clearly, $\tilde\kappa_{12} g = \tilde\kappa_1 \tilde\kappa_2 g$. 
Hence, it remains to show that $\mu_{12}=\nu$. 
We have \[ \mu_{12}(\!\dd y) = \iint   m_{1}(z,x) m_{2}(x,y) \mu(\!\dd z) \kappa_1(z, \dd x) \kappa_2(x, \dd y).   \]
This can be simplified, since
\begin{equation}\label{eq:comp m} m_1(z,x)  m_2(x,y) = \frac{g_2(x)}{(\tilde\kappa_1 g_2)(z)} \frac{g(y)}{(\tilde\kappa_2 g)(x)}=  \frac{g(y)}{(\tilde\kappa_1 \tilde\kappa_2)(z)} =\frac{g(y)}{(\tilde\kappa_{12} g)(z)} = m_{12}(z,y), \end{equation}
using $g_2(x)=(\tilde\kappa_2 g)(x)$. The second equality follows from cancellation of  the numerator of the first term and denominator of the second term. 
Hence 
\[	\mu_{12}(\!\dd y) = \iint  m_{12}(z,y) \mu(\!\dd z) \kappa_1(z, \dd x) \kappa_2(x, \dd y)= \int m_{12}(z,y)\mu(\!\dd z) \kappa_{12}(z,\dd y) = \nu(\!\dd y). \]
\end{proof}
In particular, this theorem shows compositionality of Doob's $h$-transform in the special case where $\tilde\kappa$ is taken equal to $\kappa$.

\subsection{Product of kernels}
Next, we define the parallel product of optics.
Recall the definition of the $\odot$-product from Corollary \ref{eq:product_h}. 
\begin{defn}\label{def:parallel optics}
Assume Markov kernels
$\kappa \colon S_i \rightarrowtriangle T_i$ and $\tilde\kappa \colon S_i \rightarrowtriangle T_i$. 
Denote the  {\it parallel product of the optics} $F(\kappa_1, \tilde\kappa_1)$ with $F(\kappa_2, \tilde\kappa_2)$ by $F(\kappa_1, \tilde\kappa_1) \otimes F(\kappa_2, \tilde\kappa_2)$. 
Let 
\[ \cD(T_1,T_2,S_1,S_2)= \left\{(h, \mu) \in \cB(T_1 \times T_2) \times \cM(S_1 \times S_2) \colon h = h_1\odot h_2,\, \mu = \mu_1 \otimes \mu_2\right\}. \]
Define $F(\kappa_1, \tilde\kappa_1) \otimes F(\kappa_2, \tilde\kappa_2) \colon  \cD(T_1,T_2,S_1,S_2) \to  \cD(S_1,S_2,T_1,T_2)$ by 
\begin{equation}\label{eq:parallel} \left(F(\kappa_1, \tilde\kappa_1) \otimes F(\kappa_2, \tilde\kappa_2)\right)(g_1\odot g_2, \mu_1\otimes \mu_2) = \left(g_1'\odot g_2', \forw{\kappa_1}(m_1,\mu_1) \otimes  \forw{\kappa_2}(m_2,\mu_2)\right), \end{equation}
where $(m_i, g_i')= \backw{\kappa_i}(g_i)$ ($i=1,2$). 
\end{defn}
The following figure facilitates understanding the composition rule:
\[
\begin{tikzpicture}[every node/.style={scale=0.8}]
\begin{scope}[on grid]

\node[vert] (l) at (0, 0) {$\forw{1}$};
\node[vert] (r) at (6, 0) {$\tilde{\cB}_{1}$};

\node (S) [left of=l] {$\mu_1'$};
\node (A) [below right = 2 and 2 of l] {$\mu_1$};
\node (S') [right of=r] {$g_1'$};
\node (A') [below left = 2 and 2 of r] {$g_1$};

\draw[<-] (S) -- (l);
\draw[<-] (l) to[out=south east,in=west] (A);

\draw[->] (S') -- (r);
\draw[->] (r) to[out=south west,in=east] (A');

\draw[<-] (l) to[out=north east, in=west] ++(1,0.75)
 to ++(4,0)
 to[out=east, in=north west] (r)
;

\node[vert] (l') at (0, -1.5) {$\forw{2}$};
\node[vert] (r') at (6, -1.5) {$\tilde{\cB}_{2}$};

\node (T) [left of=l'] {$\mu_2'$};
\node (B) [below right = 1 and 2 of l'] {$\mu_2$};
\node (T') [right of=r'] {$g_2'$};
\node (B') [below left = 1 and 2 of r'] {$g_2$};

\draw[<-] (T) -- (l');
\draw[<-] (l') to[out=south east,in=west] (B);

\draw[->] (T') -- (r');
\draw[->] (r') to[out=south west,in=east] (B');

\draw[<-] (l') 
 to[out=north east, in=west] ++(2,1.5)
 to ++(2,0)
 to[out=east, in=north west] (r')
;

\node[draw,dashed,fit=(A) (A') (B) (B'), inner xsep = 12pt] (box) {};

\end{scope}
\end{tikzpicture}
\]

\begin{thm}\label{thm:parcomp} On  $\cD(T_1,T_2,S_1,S_2)$
\begin{equation}\label{eq:opticsparallel} F(\kappa_1, \tilde\kappa_1) \otimes F(\kappa_2, \tilde\kappa_2) =F(\kappa_1\otimes \kappa_2,\tilde\kappa_1 \otimes \tilde\kappa_2). \end{equation} 
\end{thm}
\begin{proof}
Define
\[ (\bar g, \bar\nu) = \optic{\kappa_1\otimes \kappa_2}{\tilde\kappa_1\otimes \tilde\kappa_2}(g_1\odot g_2, \mu_1\otimes \mu_2). \]
From Definition \ref{def:parallel optics} it follows that we need to show that this equals the right-hand-side of \eqref{eq:parallel}.
First note that 
\[ \bar g = (\tilde\kappa_1\otimes \tilde\kappa_2)(g_1\odot g_2) = (\tilde\kappa_1 g_1) \otimes (\tilde\kappa_2 g_2) = g_1' \odot g_2'. \]
Now, with 
\[	m(x,y) = \frac{(g_1\odot g_2)(y)}{(\tilde\kappa_1 \otimes \tilde\kappa_2)(g_1\odot g_2)(x)} \]
we have 
\begin{align*} \bar\nu(\dd y)  &= \int m(x,y) (\kappa_1 \otimes \kappa_2)(x,\dd y) (\mu_1\otimes \mu_2)(\dd x) \\ & = \prod_{i=1}^2 \int \frac{g_i(y_i)}{(\tilde\kappa_i g_i)(x_i)} \kappa_1(x_i,\dd y_i) \mu_i(\dd x_i)= \forw{\kappa_1}(m_1,\mu_1) \otimes  \forw{\kappa_2}(m_2,\mu_2), \end{align*}
where $(m_i, g_i')= \backw{\tilde{\kappa}_i}(g_i)$ ($i=1,2$). 
\end{proof}

Contrary to the result on sequential composition of optics (Theorem \ref{thm:seqcomp}), this theorem is restricted to product measures. Hence, to forward sample in a compositional way, Theorem \ref{thm:parcomp} suggests that we first need to marginalise the measure that is pushed forward by $\cF$. 
\begin{defn}\label{def:M}
The marginalisation operator  $M$  acting on measures $\mu\in \cM(S \otimes S')$ is defined by by (using postfix notation for $M$)
\[(\mu  M)(B \times B') = \frac{ \int \mu(B, \dd x')}{\sqrt{\int \dd \mu}} \otimes \frac{\int \mu(\! \dd x, B')}{\sqrt{\int \dd \mu}},\qquad B \in \fB, B' \in \fB'.
\]
\end{defn}
Rather than pushing forward a measure using the map $\forw{}$, one can also push forward a weighted sample, which does not suffer from such restrictions. From a statistical perspective, this is the more interesting setting. To define sampling in a compositional way, we recall the following result.
\begin{lem}[\cite{kallenberg2002foundations}, p.~56]\label{lemma: randomisation}
For a Markov kernel $\kappa$ between a measure space $(E, \fB)$ and a Borel space $(E', \fB')$, there is a measurable function $\sigma_\kappa\colon E\times [0,1] \to E'$  and a random variable $Z \sim U[0,1]$ such that the random variable $\sigma_\kappa(x, Z)$ has law $\kappa(x, \cdot)$ for each $x \in E$. The random variable $Z$ is  called an \emph{innovation variable}.
\end{lem}
Then, compositionality results for sampling can be obtained by explicitly ``carrying along'' the innovation variable $Z$. A detailed exposition of this approach in the language of category theory is part of ongoing research and considered slightly out of scope for this paper. On a directed tree this yields a fully compositional algorithm, where we are free to either first compose Markov kernels and then use BFFG or to compose optics. Such structure is however lost on a directed acyclic graph which is not a tree. This can be seen from Definition \ref{defn:pullback_dag}, where the pullback is defined as an operator, which not necessarily needs to be the pullback of a Markov kernel.

\subsection{Automatic BFFG}

We now have  all constituents for a universal procedure of backward filtering and forward guiding (by sampling), which doesn't put constraints on the program architecture, because it is applicable to any composition of (product-) kernels
and it preserves the program architecture. The transformed program is a two-pass algorithm where  the forward pass has the same architecture as the program: it has the same Markovian structure, the same function calls, the same diagrammatical form. 
Also the backward pass has the same architecture, though traversed backwards.

Forwards, a kernel $\kappa$ is replaced by $\forw\kappa$, while backwards it is replaced by $\backw{\tilde\kappa}$. Jointly, they form the optic $F(\kappa, \tilde\kappa) = \optic{\kappa}{\tilde\kappa}$.  As an example, automatic BFFG entails that the forward evolution
\begin{equation}\label{eq:forwardexample} \delta_{x_0} \kappa_1 \Dup (\kappa_2\otimes \kappa_3) \kappa_4 \end{equation}
is turned into 
\begin{equation}\label{eq:progtransform1}  \delta_{x_0} F(\kappa_1, \tilde\kappa_1) F(\Dup, \Dup) F (\kappa_2 \otimes \kappa_3, \tilde\kappa_2 \otimes \tilde\kappa_3) 
\end{equation}
by applying $F$ to all terms except for the final kernel pointing to an observation leaf. Kernel $\kappa_4$ emits information in the form of $g_3$, defined by  
\[ g_3 \colon  x \mapsto   \left.\frac{\dd \kappa_4(x, \dd y)}{\lambda(\dd y)}\right|_{y=x_{\rm obs}}. \]
By applying the program transformation \eqref{eq:progtransform1} to $(g_3, \delta_{x_0})$ a weighted sample of the smoothing distribution can be obtained, as in \eqref{eq:weighted_sample}. Equation \eqref{eq:progtransform1} is a translation of Markov kernels in \eqref{eq:forwardexample} to optics.

In this way, for $X$ on any DAG (translated to the corresponding string diagram of kernel compositions) an approximation $X^\circ$ of the conditional $X^\star$ is obtained that can be sampled from. {\it Note that  $X^\circ$, unlike $X^\star$, preserves the Markovian structure of $X$ and does not introduce new edges in the dependency graph (as it was the case in Figure~\ref{fig:collider}). This is of importance: for a probabilistic program which is able to sample $X$, it is easier  to adapt it to sample from $X^\circ$ with weights, than to adapt it to sample from $X^\star$ directly (which has a different dependence structure). The former is possible as automatic program transform.}

\section{Examples of discrete-time guided processes}\label{sec:examples}

In this section we provide various examples to illustrate our approach. 

The examples make use of the idea that the composition of backward steps is computationally cheap if it preserves a parametrised form of $g$.  Hence,  suppose $g$ can be parametrised by $g=U(\zeta)$ for a parameter $\zeta$ and a parametrisation $U$. Backward filtering is tractable when
\begin{itemize}
	\item $\tilde \kappa g= U(\zeta')$ for some $\zeta'$;
	\item $\kappa_{\Dup_k}(g_1\odot \cdots \odot g_k) =U(\zeta')$ for some $\zeta'$ (this corresponds to the fusion step).
\end{itemize}
In some models the (otherwise intractable) fusion step can be bypassed by only applying the backward map on a line graph, see 
Section \ref{sec:backward diagonlisation} for examples. In case $g$ is proportional to the density of an exponential family member distribution, then the fusion step  is tractable.

In applying the forward map for sampling less tractability is required, as for example the pseudo-marginal Metropolis-Hastings algorithm can be used to unbiasedly estimate $(\kappa g)(x)$. Details are given ahead in Section \ref{subsec:mcmc}.

\subsection{Gaussian filtering and smoothing}\label{subsec:gauss}

Consider the stochastic process on $\cG$ with transitions defined by 
\begin{equation}\label{eq:gaussian_model}X_t \mid X_s = x \sim N(\mu_t(x), Q_t(x)), \qquad t\in \ch(s).  \end{equation}
Unfortunately, for this class of models (parametrised by $(\mu_t, Q_t)$), only in very special cases Doob's $h$-transform is tractable. For that reason,  we look for  a tractable $h$-transform to define a guided process $X^\circ$. This is obtained in the specific case where 
\[\tilde X_t \mid \tilde X_s =x \sim N(\Phi_t x +  \beta_t, Q_t), \qquad t  \in \ch(s). \]
Below we show that backward filtering, sampling from the forward map and computing  the weights are all  fully tractable.

In the following, we write $\phi(x; \mu, \Sigma)$ for the density of the $N(\mu,\Sigma)$-distribution, evaluated at $x$. Similarly, we write $\phi^{\rm{can}}(x; F, H)$ for the density of canonical Normal distribution with potential $F=\Sigma^{-1} \mu$ and precision matrix $H = \Sigma^{-1}$, evaluated at $x$.  

The $g$ functions  can be parametrised by the triple $(c, F, H)$:\begin{equation}\label{eq:gaussian_h} \begin{split} g(y) &= \exp \left(c -\frac12 y^\T H y + y^\T F \right) \\ & =  \varpi(c,F,H) \phi^{\rm{can}}(y; F, H)= \varpi(c, F, H) \phi(y; H^{-1}F, H^{-1}), \end{split} \end{equation}
where 
\[ 	\log \varpi(c, F, H) = c - \log \phi^{\on{can}}(0, F, H).\]
We write $g(x)=U(c, F, H)(x)$ for $g$ and parameters as in \eqref{eq:gaussian_h}.

\begin{thm}\label{thm:gaussianformulas}
Let $\kappa(x,\dd y) = \phi(y; \mu(x), Q(x)) \dd y$ and $\tilde\kappa(x,\dd y) = \phi(y; \Phi x + \beta, Q)\dd y$.  Assume that $g=U(c, F, H)$ with invertible $H$. 
\begin{enumerate}[label=(\roman*)]
\item Pullback for $\kappa$: With $C(x)= Q(x)+ H^{-1}$ invertible,
\[ (\kappa g)(x) = \varpi(c, F, H) \phi(H^{-1}F; \mu(x), C(x)).
\]
\item Backward Information Filter initialisation from leaves:  if $y \sim N(\Phi x + \beta, Q)$ is observed at a leaf, then 
\[
 g(x) =  U(\bar c, \bar F, \bar H)(x)\quad\text{where}\quad
\begin{cases} \bar H = \Phi^\T Q^{-1} \Phi \\
 \bar F = \Phi^\T Q^{-1} (y - \beta)\\
	\bar c =   \log  \phi(\beta; y, Q).
\end{cases}
\]
	\item Pullback for $\tilde \kappa$: With  $C= Q+ H^{-1}$ invertible,
\[
  \tilde \kappa g = U(\bar c, \bar F, \bar H)	
  \quad\text{where}\quad
\begin{cases}
	\bar H = \Phi^\T C^{-1} \Phi \\
\bar F = \Phi^\T C^{-1} (H^{-1}F-\beta)\\
	\bar c =  c   - \log \phi^{\rm can}(0, F, H) + \log  \phi(\beta; H^{-1}F, C)
\end{cases}
\]
and
\[\backw{\tilde \kappa} =  (m, \tilde \kappa g), \quad  m(x,y) = \frac{U(c, F, H)(y)}{ U(\bar c, \bar F, \bar H)(x)}.\]

If $\Phi$ is invertible, then alternatively
\[	\log \varpi(\bar c, \bar F, \bar H) = \log \varpi(c, F, H)  - \log |\Phi|. \]

\item Forward map: 
\[ \forw{\kappa}(m, \mu) = \nu, \quad \nu(\!\dd y) = \int w(x) \phi^{\rm can}(y; F+ Q(x)^{-1}  \mu(x), H+ Q(x)^{-1}) \mu(\!\dd x) \dd y, \]
where the weight at $x$ is given by  \[ w(x) = \frac{(\kappa g)(x)}{(\tilde\kappa g)(x)} =  
 \frac{ (\kappa g)(x)}{ U(\bar c, \bar F, \bar H)(x)}
=    \frac{\phi(H^{-1}F; \mu(x), Q(x) + H^{-1})}{\phi^{\on{ can}}(0,F,H) U(\bar c - c, \bar F, \bar H)(x)}.\]

In particular, if ${\mu=\varpi \delta_x}$, the forward map gives the importance weight and conditional 
distribution for the next random draw, 
\[
\varpi w(x) \cdot N^{\mathrm{can}}\left( F + Q(x)^{-1} \mu(x), H + Q(x)^{-1} \right).
\]

\item Fusion: if $g_i = U(c_i, F_i, h_i)$, $i = 1, \dots, k$, then
\[  \Dup_k(g_1 \odot \cdots \odot g_k) =  U\left( \sum_{i=1}^k c_i, \sum_{i=1}^k F_i, \sum_{i=1}^kH_i\right). \]
In particular, $\forw{\Dup_k}(\cdot, \varpi \delta_x)$, where $\varpi$ is a weight, makes $k$ weighted copies $\varpi^{1/k} \delta_x$

\item 

Backward splitting of $g$ to two parents nodes: suppose $g(y) = \varpi(c,F,H) \phi(y; \mu, P)$, with $\mu=H^{-1}F$ and $P=H^{-1}$ and assume a flat prior on the parent vertices. Consider the partitioning
\[ y=\Bm y^{(1)} \\ y^{(2)} \Em \qquad  \mu=\Bm \mu^{(1)} \\ \mu^{(2)} \Em \qquad P=\Bm P^{(11)} & P^{(12)}\\ P^{(21)} & P^{(22)} \Em. \]
Then
\[ (\bar\kappa g)(y^{(1)}, y^{(2)}) = \sqrt{\varpi(c,F,H)}\phi(y^{(1)}; \mu^{(1)}, P^{(11)}) \odot \sqrt{\varpi(c,F,H)}\phi(y^{(2)}; \mu^{(2)}, P^{(22)}). \]

\end{enumerate}
\end{thm}

\begin{proof} The proof is elementary. For completeness, the proof is given in  appendix \ref{sec:proofs}. For the pullback, similar results have appeared in the literature,  for example  Chapter 7 in \cite{chopin2020introduction},   \cite{wilkinson2002conditional} and Chapter 5 in \cite{cappe2005springer}. 
\end{proof}
Taking an inverse of $H$ respective $C$ can be avoided using
$
(Q+H^{-1})^{-1} H^{-1} = (Q H+I)^{-1} , 
(Q+H^{-1})^{-1}  = H - H (H+Q^{-1})^{-1} H .  
$
While backward filtering for a linear Gaussian process on a tree is well known (see e.g.\ \cite{Chou1994}, section 3), results presented  in the literature often don't state update formulas for the constant $\varpi$ (or equivalently $c$). If the mere goal is smoothing with known dynamics, this suffices. However, we will also be interested in estimating parameters in $\mu$ and or $Q$ then the constant cannot be ignored.

In case the dynamics of $X$ itself are linear, it simply follows that the weights equal $1$. Moreover, if additionally the process lives on a ``line graph with attached observation leaves'', where each non-leaf vertex has one child in $\cS$ and at most one child in $\cV$, our procedure is essentially equivalent to Forward Filtering Backward Sampling (FFBS, \cite{CarterKohn(1994)}), the difference being that our procedure applies in time-reversed order. On a general directed tree however the ordering cannot be changed and the filtering steps must be done backwards, as we propose, just like in message passing algorithms in general.

\newcommand{\I}[1][]{[{#1}]}

\subsection{Discrete state-space Markov chains and particle systems}\label{sec:discrete}

\subsubsection{Branching particle on a tree}
Assume a  ``particle'' takes values in a finite state space $E=\{1,\ldots, R\}$ according to the  $R\times R$ transition matrix $K$.  At each vertex, the particle is allowed to copy itself and branch. 

The algorithmic elements  can easily be identified. The forward kernel $\kappa(x,\dd y)$ can be identified with the  matrix $K$. Furthermore, the map $x\mapsto g(x)$ can be identified with the column vector $\bs{h}=[g_1,\ldots, g_R]^\T$, where $g_i=g(i)$. 
Hence $g$ is parametrised by $\bs{g}$ and we write $g=U(\bs{g})$.

The following theorem is easily proved. It gives the necessary ingredients for applied the backward and forward map. 
\begin{thm}
Let $\kappa(x,\dd y)$ and d $\tilde\kappa(x,\dd y)$ be represented by the stochastic matrices $K$ and $\tilde K$ respectively.  Assume that $h=U(\bs{g})$. 
\begin{enumerate}[label=(\roman*)]
\item Pullback for $\kappa$: $\kappa g = U(K\bs{g})$. 
	\item Pullback for $\tilde \kappa$: $\tilde \kappa g = U(\tilde  K\bs{g})$ and
\[\backw{\tilde \kappa} =  (m, \tilde \kappa g), \quad  m(x,y) = \frac{U(K g)(y)}{ U(\tilde K g)(x)}.\]
 Denote by $M$ the matrix with elements $m(x,y)$, $x,y \in E$.

\item Fusion: if $g_i = U(\bs{g}_i)$, $i = 1, \dots, k$, then
\[  \Dup_k(g_1 \odot \cdots \odot g_k) =  U\left( \bigcirc_{i=1}^k \bs{g}_i\right),\] where $\bigcirc$ denotes the Hadamard (entrywise) product.
\item Forward map:  Let $\mu$ be a (row) probability vector. If $\forw{\kappa}(m, \mu)=\nu$, 
then for $k\in E$
\[	 \nu(\{k\}) = \mu (M\bigcirc K) e_k, \] with $e_k$ the $k$-th standard basis-vector in $\RR^R$. 

\item The weight at $x\in E$ is given by  \[ w(x) = \frac{(\kappa g)(x)}{(\tilde\kappa g)(x)} =   \frac{\langle K\bs{g}\rangle_x}{\langle\tilde K\bs{g}\rangle_x} \]
where we denote the $j$-th element of the vector $a$ by $\langle a\rangle_j$.

\item Initialisation at the leaves: at a leaf vertex with observation $k\in E$, set $\bs{g}= e_k$. 
\item Backward splitting of $g$ to two parents nodes:  suppose $E = \{(x,y)\colon x = 1, \dots, R_1, y = 1, \dots, R_2\}$.
Then $g$ has the form
\[ g = [g_{11} \cdots g_{1R_2} g_{21} \cdots g_{2R_2} \cdots \cdots g_{R_1 1}\cdots g_{R_1R_2}]^\T. \]
Let \[ g^{(1)}_i = \sum_{j=1}^{R_2} g_{ij} \in \bB(\{1, \dots, R_1\}), \qquad g^{(2)}_i = \sum_{j=1}^{R_1} g_{ji} \in  \bB(\{1, \dots, R_2\})\]
and $c =\sum_{i,j} h_{ij}$, then under  a flat prior on the parent vertices
$\bar\kappa g = g^{(1)}/ \sqrt c \odot g^{(2)}/  \sqrt c $. 

\end{enumerate}
\end{thm}

While we can simply take $\tilde\kappa=\kappa$ on a tree, yielding unit weights, in case $R$ is very large, it can nevertheless be advantageous to  use a different (simpler) map $\tilde\kappa$. To see, this, note that we only need to compute one element of the matrix vector product $Kg$, but need to compute the full vector $\tilde{K}g$ in the backward map. Hence, choosing $\tilde K$ sparse can give computational advantages.

\subsubsection{Interacting particles -- cellular automata -- agent based models}

Now consider a discrete time interacting particle process, say with $n$ particles, where each particle takes values in $\{1,\ldots, R\}$. Hence, a particle configuration $x$ at a particular time-instant takes values in $E=\{1,\ldots, R\}^n$.

A simple example consists of particles with state $E = \{1\equiv \mathbf{S}, 2\equiv \mathbf{I}, 3\equiv \mathbf{R}\}$ that represent the health status of individuals that are either {\bf S}usceptible, {\bf I}nfected or {\bf R}ecovered. Then the probability to transition from state {\bf S} to {\bf I} may depend on the number of nearby particles (individuals) that are infected, hence the particles are ``interacting''. A similar example is obtained from time-discretisation of the contact process (cf.\ \cite{Liggett05}).
A visualisation of the DAG is given in Figure~\ref{fig:dag sir}: time moves from left to right, vertical connects represent state changes of particles and the diagonal connections represent local interactions. Note that each particle has multiple parents defined by a local neighbourhood and that a model like this is sometimes referred to as a cellular automaton.

Let $x\in E$. Without additional assumptions, the forward transition kernel 
can be represented by a $R^n \times R^n$ transition matrix. With a large number of particles such a model is not tractable computationally. To turn this into a more tractable form, we make the simplifying assumption that  conditional on $x$ each particle transitions independently. This implies that  the forward evolution kernel $\kappa(x,\dd y)$ can be represented by $(K_1(x),\ldots, K_n(x))$, where $K_i(x)$ is the $R\times R$ transition matrix for the $i$-th particle. Note that the particles interact because the transition kernel for the $i$th particle  may depend on the state of {\it all} particles. 

Due to interactions, it will computationally be very expensive to use $\kappa$ for the backward map. Instead, in the backward map we propose to ignore/neglect all interactions between particles. This means that effectively we use a backward map where all particles move independently, and the evolution of the $i$-th particle  depends only on $x_i$ (not $x$, as in the forward kernel). Put differently, in the backward map, we simplify the graphical model to $n$ line graphs, one for each particle. See Figure~\ref{fig:dag sir}. 

This choice implies that  $\tilde\kappa$ can be represented by $(\tilde K_1,\ldots, \tilde K_n)$ and the map $x\mapsto g(x) = \prod_{i=1}^n g_i(x)$ can be represented by $(\bs{g}_1,\ldots,\bs{g}_n)$, where  $\bs{g}_i = [g_{i1},\ldots, g_{iR}]^\T$. We write $g=U(\bs{g}_1,\ldots,\bs{g}_n)$.

\begin{thm}
Assume $\kappa$ and $\tilde \kappa$ are represented by $R\times R$ stochastic matrices $K_1,\ldots, K_n$ and $\tilde K_1,\ldots, \tilde K_n$ respectively.  Assume that $g=U(\bs{g}_1,\ldots,\bs{g}_n)$.
\begin{enumerate}[label=(\roman*)]
\item Pullback for $\kappa$: $U(\kappa g)(x) = (K_1(x)\bs{g}_1,\ldots,K_n(x)\bs{g}_n)$

	\item Pullback for $\tilde \kappa$: $U(\tilde \kappa g) = (\tilde K_1\bs{g}_1,\ldots,\tilde K_n\bs{g}_n)$
\item Parametrisation of messages: For each $i\in \{1,\ldots, n\}$ let $M_i$ be the matrix with elements
\[ M_i(x,y) = \frac{U(K_i g)(y)}{ U(\tilde K_i g)(x)}.\]
Set $m = (M_1,\ldots, M_n)$. 
\item 
 Forward map:  Let $\mu_1,\ldots, \mu_n$ be a $n$ (row) probability vectors with $\mu_i$ representing the distribution of the $i$-th particle. If $\mu = \otimes_{i=1}^n \mu_i$ and $\forw{\kappa}(m, \mu)=\nu$, then $\nu = \otimes_{i=1}^n \nu_i$ where 
 for $k\in E$
\[	 \nu_i(\{k\}) = \mu (M_i\bigcirc K_i(x)) e_k, \] with $e_k$ the $k$-th standard basis-vector in $\RR^R$.

\item The weight at $x$ is given by \[ w(x) = \frac{(\kappa g)(x)}{(\tilde \kappa g)(x)} = \prod_{i=1}^n \frac{\langle K_i(x) \bs{g}_i\rangle_{x_i}}{\langle\tilde K_i \bs{g}_i\rangle_{x_i}}. \]

\item Initialisation at the leaf: for observation $x_{M+1}\in E$ set $\bs{g}_i = e_{\langle x_{N+1}\rangle_i}$ for $i \in \{1,\ldots, n\}$. 
\end{enumerate}
\end{thm}

\subsection{Backward diagonalisation}\label{sec:backward diagonlisation}
The example of the previous section shows a generic way to deal with interacting particles with a pattern as shown in Figure \ref{fig:dag sir}. 
Here, conditional on the state of all particles at a particular ``time'', all particles transition independently, though with transition probabilities that may depend on the state of {\it all} particles. The backward filtering is however done on $n$ separate line graphs, thereby fully bypassing the need of a tractable fusion step. We call this backward diagonalisation. 

In case of a finite-state space, the fusion step can be done, but the computational cost is exponentially increasing with the number of time-steps.
The next section provides an example where the fusion step is not tractable at all, but backward diagonalisation ensures it is not needed.

\begin{figure}
\begin{center}
\begin{tikzpicture}[style={scale=0.5}]
	\tikzstyle{empty}=[fill=white, draw=black, shape=circle,inner sep=1pt, line width=0.7pt]
	\tikzstyle{solid}=[fill=black, draw=black, shape=circle,inner sep=1pt,line width=0.7pt]
	\begin{pgfonlayer}{nodelayer}
		\node [style=solid,label={$1(0)$}] (0) at (-5, 5) {};
		\node [style=solid,label={$1(1)$}] (1) at (-3, 5) {};
		\node [style=solid,label={$1(3)$}] (2) at (-1, 5) {};
		\node [style=solid,label={$2(0)$}] (3) at (-5, 3) {};
		\node [style=solid] (4) at (-3, 3) {};
		\node [style=solid] (5) at (-1, 3) {};
		\node [style=solid] (6) at (-5, 1) {};
		\node [style=solid] (7) at (-3, 1) {};
		\node [style=solid] (8) at (-1, 1) {};
		\node [style=solid,label={$n(0)$}] (9) at (-5, -5) {};
		\node [style=solid] (10) at (-3, -5) {};
		\node [style=solid] (11) at (-1, -5) {};
		\node [style=none] (12) at (-5, -1) {};
		\node [style=none] (13) at (-3, -1) {};
		\node [style=none] (14) at (-1, -1) {};
		\node [style=none] (15) at (-5, -3) {};
		\node [style=none] (16) at (-3, -3) {};
		\node [style=none] (17) at (-1, -3) {};
		\node [style=none] (18) at (1, 5) {};
		\node [style=none] (19) at (1, 3) {};
		\node [style=none] (20) at (1, 1) {};
		\node [style=none] (21) at (1, -3) {};
		\node [style=none] (22) at (1, -5) {};
		\node [style=none] (23) at (1, -1) {};
	\end{pgfonlayer}
	\begin{pgfonlayer}{edgelayer}
		\draw [style=dotted] (0) to   (4);
		\draw [style=edge] (0) to (1);
		\draw [style=edge] (1) to (2);
		\draw [style=edge] (4) to (5);
		\draw [style=dotted] (4) to  (2);
		\draw [style=dotted] (1) to  (5);
		\draw [style=dotted] (3) to  (1);
		\draw [style=edge] (3) to   (4);
		\draw [style=dotted] (6) to  (4);
		\draw [style=dotted] (3) to  (7);
		\draw [style=edge] (7) to (8);
		\draw [style=dotted] (7) to  (5);
		\draw [style=dotted] (4) to (8);
		\draw [style=edge] (9) to (10);
		\draw [style=edge] (10) to   (11);
		\draw [style=edge] (6) to (7);
		\draw [style=dotted] (9) to (16.center);
		\draw [style=dotted] (15.center) to (10);
		\draw [style=dotted] (10) to (17.center);
		\draw [style=dotted] (8) to (20.center);
		\draw [style=dotted] (5) to (20.center);
		\draw [style=dotted] (5) to (19.center);
		\draw [style=dotted] (5) to (18.center);
		\draw [style=dotted] (2) to (18.center);
		\draw [style=dotted] (2) to (19.center);
		\draw [style=dotted] (8) to (19.center);
		\draw [style=dotted] (12.center) to (7);
		\draw [style=dotted] (6) to (13.center);
		\draw [style=dotted] (13.center) to (8);
		\draw [style=dotted] (7) to (14.center);
		\draw [style=dotted] (16.center) to (11);
		\draw [style=dotted] (8) to (23.center);
		\draw [style=dotted] (11) to (21.center);
		\draw [style=dotted] (11) to (22.center);
	\end{pgfonlayer}
\end{tikzpicture}
\end{center}
\caption{DAG for the interacting particle process of section \ref{sec:discrete}. Particles move from left to right and the grey grid depicts the dependence structure. Here, parents of a vertex are defined by a local neighbourhood (in this case all particles at distance at most 1). For example particle $1$ at time $1$ (indexed by $1(1)$) has parents $1(0)$ and $2(0)$. For backward filtering, the dynamics are reduced }
\label{fig:diagonalise}\label{fig:dag sir}

\end{figure}
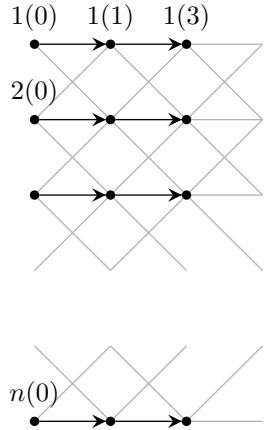

\subsection{Independent Gamma increments}\label{subsec:Gamma_example}

In this example, we consider a process with Gamma distributed increments. We consider a line graph (where $|\ch(s)| = 1$ for $s \in \cS_0$) with a single observational leaf $v$, attached to the end of the single path in $\cS_0$.
We write $Z \sim \on{Gamma}(\alpha,\beta)$ if $Z$ has density $\psi(x; \alpha, \beta) =  \beta ^{\alpha }\Gamma (\alpha )^{-1} x^{\alpha -1}e^{-\beta x} \ind_{ (0,\infty)}(x)$.  

Choose  mappings $\alpha_t$,  $\beta_t\colon (0,\infty) \to [1/C, C]$ for some $C>0$. We define a Markov process  with Gamma increments on $\cG$  by 
\[   	X_{t} - X_s \mid X_s=x \sim  \on{Gamma}(\alpha_t, \beta_t(x)), \quad \text{if $s = \pa(t)$}, \qquad  X_0 = x_0. \] 
This implies
\[ \pedge{}{t}{x}{y}  = \psi(y - x; \alpha_t, \beta_t(x)). \]
 Note that $X$ can be thought of as a time discretised version of the  SDE $\dd X_t = \beta^{-1}(X_t)\dd L_t$ driven by a Gamma  process  $(L_t)$ with scale parameter 1, observed at final time $T$.  See \cite{Belomestnyetal2019} for a continuous time perspective on this problem. A statistical application will typically involve multiple observations and henceforth multiple line graphs. Simulation on each line graph corresponds to conditional (bridge) simulation for the process $X$.

As the $h$-transform is not tractable, we introduce the process $\tilde{X}$ which is defined as  the Markov process  with Gamma increments  on $\cG$ induced by  $(\alpha_t, \beta)$,  with $\beta$ a user specified nonnegative constant. This process is tractable and is used to define the $h$-transform.

As before, we state our results using the kernels $\kappa$ and $\tilde\kappa$. Define for $A, \beta>0$
\begin{equation}\label{eq:h_gammaexample} g(y) = \psi(x_v-y; A, \beta) \end{equation}
and write $h=U(A, \beta)$.

 Define the exponentially-tilted Beta-distribution with parameters $\gamma_1, \gamma_2>0$ and $\lambda\in \RR$ as the distribution with density
\begin{equation}\label{eq:exptilted_beta} q_{\gamma_1, \gamma_2, \lambda}(z) \propto z^{\gamma_1-1} (1-z)^{\gamma_2-1} e^{-\lambda z} \ind_{(0,1)}(z). \end{equation}
We denote this distribution by $\mathrm{ExpBeta}(\gamma_1, \gamma_2, \lambda)$. Sampling from this distribution can be accomplished for example  using rejection sampling, with importance sampling distribution Beta($\gamma_1, \gamma_2$).

\begin{thm}\label{lem:gamma_incr_htilde_update}
Let $\kappa(x,\dd y) = \psi(y-x; \alpha, \beta(x))\dd y $ and $\tilde\kappa(x,\dd y) = \psi(y-x; \alpha, \beta)\dd y $.  Assume that  $g=U(A, \beta)$. 
\begin{enumerate}[label=(\roman*)]
\item Pullback for $\kappa$: 
\[ (\kappa g)(x) = U(A+\alpha, \beta)(x)  \left(\frac{\beta(x)}{\beta}\right)^{\alpha} \: \EE\, e^{-\xi(x) Z}, \]
where $Z \sim \mathrm{Beta}(\alpha, A)$ and $\xi(x)=(\beta(x)-\beta)(x_v-x)$. 

\item Pullback for $\tilde \kappa$: $(\tilde \kappa g)(x)=U(A+\alpha, \beta)(x)$.

\item Forward map: if $\mu = \delta_x$, then a draw $y$ from $\forw{\kappa}(m, \mu)$ can be generated by first drawing $z \sim \mathrm{ExpBeta}(\alpha, A, \xi(x))$ and then setting $y = x+ z(x_v-x)$. 
\item The weight at $x$ is given by  \[ w(x) = \frac{(\kappa g)(x)}{(\tilde\kappa g)(x)} =  \left(\frac{\beta(x)}{\beta}\right)^{\alpha} \EE\, e^{-\xi(x) Z}\]
 where $Z \sim \mathrm{Beta}(\alpha, A)$
\end{enumerate}
\end{thm}

\begin{proof} See Appendix \ref{sec:proofs}.\end{proof}

Note that the expression for the weight immediately suggests a method to estimate $w(x)$ unbiasedly. 
The setting is restricted to a line graph, as  fusion of $g_1$ and $g_2$ of the form \eqref{eq:h_gammaexample} does not  lead to a fused function of the same form.

The model can be extended to $n$ Markov processes with Gamma increments, which evolve conditionally independent, where it is assumed that the forward evolution consists of composing kernels
 \[ \kappa(\bs{x}, \dd \bs{y})  = \prod_{i=1}^n  \psi(\bs{y}_i - \bs{x}_i; \alpha, \beta(\bs{x})) \dd \bs{y}. \]	
 The backward kernel $\tilde \kappa$ is then taken to be the same, with $\beta(\bs{x})$ replaced with $\beta$. This is an instance of backward diagonalisation.

\begin{rem}
The setting of this section can be generalised: suppose for $s=pa(t)$ that 
\[
X_t - X_s \mid X_s = x \sim \cD(\theta_t(x)),
\]
with $\cD(\th)$ a distribution with parameter $\th$ supported on $(0,\infty)$, say with density $\psi(\cdot; \theta)$.

Tracing back the proof, the derivation of $\kappa h$ shows that if $\kappa(x,\dd y) = \psi(y-x; \theta)\dd y$, then backward filtering is tractable if
\[ (\kappa g)(x) = \int \psi(x_v-y; \theta) \psi(y-x; \theta(x)) \dd y = \int \psi(x_v-x-z; \theta) \psi(z;\theta(x)) \dd z \]
can be written as $c \psi(x_v-x; \tilde\theta(x))$ for some $\tilde\theta(x)$. Hence, the results can be extended to the case where  increments are closed under convolution. The Gamma distribution has this property,  but so does the Inverse Gaussian $\on{IG}(\mu, 2\mu^2)$-distribution for example.
	
\end{rem}

\section{Examples of  continuous-time guided processes}\label{sec:examples_continous}

\subsection{Stochastic differential equations}
Suppose the Markov process $X$ is defined as a solution to the stochastic differential equation (SDE) \[\dd X_t = b(t,X_t) \dd t + \sigma(t,X_t) \dd W_t.\] 
Here $\sigma$ is possibly non-constant, as is often the case for physically meaningful local noise structure, for example for those obtained through the stochastically constrained variational principle, Cf.\ \cite{holm2015variational}.

Assume we take $g$ derived from the simpler SDE  $\dd \tilde X_t = \tilde{b}(t,\tilde X_t) \dd t + \tilde{\sigma}(t,\tilde{X}_t) \dd W_t$. In case the drift is linear and diffusivity only time-dependent, then $g$ is proportional to a Gaussian density and can be represented as $g(t,x) = \exp\left(c(t) + F(t)^\T x -\frac12 x^\prime H(t) x\right)$, for scalar-valued $c$, vector-valued $F$ and matrix-valued $H$ (just as in \eqref{eq:gaussian_h}). In this case the term in the exponent of the weight in Theorem \ref{thm:continuous_time_weight} satisfies
\[  \frac{({\cL}-\tilde{{\cL}})  g}{g}  = \sum_i (b_i-\tilde{b}_i) \frac{\partial_i  g}{g} + \frac12 \sum_{i,j} (a_{ij}-\tilde{a}_{ij}) \frac{\partial^2_{ij}  g}{g}, \]
where we have omitted arguments $(s,x)$ from the functions and $\partial_i$ denotes the partial derivative with respect to $x_i$. If we set $r_i = (\partial_i  g)/g$, then $(\partial_{ij}  g)/  g = \partial_j  r_i +  r_i  r_j$ and therefore
\[ \frac{({\cL}-\tilde{{\cL}})  g}{g} = \sum_i (b_i-\tilde{b}_i) r_i + \frac12 \sum_{i,j} (a_{ij}-\tilde{a}_{ij}) \left( \partial_j  r_i +  r_i  r_j\right). \]
This expression coincides with the expression for the likelihood given in Proposition 1 of \cite{schauer2017guided}, but the proof given here is much shorter. 

In case the process $\tilde X$ defines a linear SDE, then solving the Kolmogorov backwards equation on $[0,T]$ reduces to solving a system of ordinary differential equations. Theorem 2.5 in \cite{mider2021continuous} gives these ODEs in their ``information filter'' form. The output of these equations is fully compatible with the form of $g$ in Subsection \ref{subsec:gauss}. In particular, fusion is tractable and our approach enables to simulate guided processes for  conditioned  branching diffusions. 

The corresponding rules for
 \texttt{Mitosis.jl} are provided by \texttt{MitosisStochasticDiffEq.jl} using solvers from Julia's SciML (\url{sciml.ai}). ecosystem for Scientific Machine Learning \citep{rackauckas2017adaptive}.

 \subsection{First hitting time of a diffusion}
Doob's $h$-transforms are not restricted to  conditioning on  the value of a process at a fixed time. Rather, they can generally be applied for different types of conditioning such as for example conditioning on the first hitting time of a set. Accordingly,  guided processes can be used for general conditionings. 

As a concrete example, suppose that over an edge a one-dimensional diffusion process instead of evolving for a fix time, evolves until it first hits level $v$, $v$ being considered fixed and the hitting time is observed.
Suppose the dynamics of the diffusion are governed by the SDE $\dd X_t = b(X_t) \dd t + \dd W_t$. Let $\tau_v$ denote the first hitting time of $v$. Now assume existence of the density $p$ such that  $\P_{s,x}(\tau_v \in \dd T)= p(s,x; T)  \dd T$. In general $p$ is intractable, though if  $b\equiv 0$ then \[
g(t, x) := p(t,x; T)=  \frac{|v-x|}{2\pi (T-t)^3}e^{-(v-x)^2/(2(T-t))}
\] 
which can be used to define the guided process.
It is easily verified that $h$ is space-time harmonic, for $x\neq v$
\[
\pder{t}  g(t, x) + \frac12 \pder[^2]{x^2} g(t, x) = 0.
\]
Furthermore
$\nabla \log  g(t, x) =  (v-x)/(T-t) + 1/(x-v)$,
which implies 
\[
\cL^\circ f(x) =  \left(b(x) + \frac{v-x}{T-t} + \frac{1}{x-v}\right)f'(x) +  \frac12 f''(x).
\]
Note that in case $b\equiv 0$, this is the generator of a (shifted and scaled) Bessel bridge.

In the larger picture of continuous processes embedded into graphical models, the observation of such a hitting time may
correspond to an ``observation operation'' (a leaf in the DAG) with transition density $p(s, x; T)$.

\subsection{Continuous time Markov chains  on a countable set}
Let $X$ denote a continuous time Markov process taking values in  the countable  set $E$.  Let $\cQ = (q(x,y),\, x, y \in E)$ be a $Q$-matrix, i.e.\ its elements $q(x,y)$ satisfy $q(x,y) \ge 0$ for all $x\neq y$ and $\sum_y q(x,y) =0$. Define $c(x)=-q(x,x)$. If $\sup_x c(x)<\infty$, then $\cQ$ uniquely defines a continuous time Markov chain with values in $E$. Its transition probabilities $p_t(x,y)$ are continuously differentiable in $t$ for all $x$ and $y$ and satisfy Kolmogorov's backward equations:
\[ \frac{\dd}{\dd t} p_t(x,y) = \sum_x q(x,z) p_t(z,y) \]
(cf.\ \cite{Liggett2010}, Chapter 2, in particular Corollary 2.34). The infinitesimal generator  acts upon functions $f\colon E \to \RR$ by
\begin{equation}\label{eq:genCTTMchain} \scr{L} f(x) = \sum_{y \in E} q(x,y) \left(f(y) - f(x)\right) = \sum_{y\in E} q(x,y) f(y),\qquad t\in [0,\infty), \quad x \in E. \end{equation}
In the specific setting where $E$ is finite with states labeled by $1, 2, \ldots, R$, $\cQ$ is an $R\times R$-matrix. Then, by evaluating $\scr{L}f(x)$ for each $x\in E$ we can identify  $\scr{L} f$ by $\cQ \bs{f}$,
where $\bs{f}$ is the column vector with $i$-th element $\bs{f}_i =f(i)$. 

Suppose an  $h$ transform $h\colon [0,\infty) \times E \to \RR$ is specified. It induces a guided process $X^\circ$ for which the generator of the space-time process can be derived using Equation \eqref{eq:Lcirc} and the preceding display. Some simple calculations reveal that for $f\colon [0,\infty) \times E \to \RR$
\[  h(t,x) (\cA^\circ f)(t,x) = \sum_{y \in E} q(x,y) \left( f(t,y)-f(t,x)\right)  h(t,y) +  h(t,x) \frac{\partial f}{\partial t}(t,x) \]
This implies that the generator of $X^\circ$ acts upon functions $f\colon [0,\infty) \times  E \to \RR$ as
\begin{equation}\label{matrixcirc}
(\cL^\circ_t f)(x)= \sum_{y \in E} q(x,y) (f(y)-f(x))  \frac{h(t, y)}{h(t,x)}.
\end{equation}

Therefore, comparing with \eqref{eq:genCTTMchain}, we see that $X^\circ$ is a  time-inhomogenuous Markov process $X^\circ$ with time dependent generator matrix $\cQ_t^\circ = (q_t^\circ(x,y),\, x, y \in E)$, where
\[ q_t^\circ(x,y) =  q(x,y)  \frac{h(t, y)}{h(t,x)} \quad \text{if} \quad y\neq x\] 
and $q_t^\circ(x,x) = 1- \sum_{y\neq x} q^\circ_t(x,y)$.

Now assume that $\tilde X$ is a  continuous time Markov process on $E$ with tractable $h$-function $g$ such that  $\tilde{\cA} g=0$, where $\tilde\cA$ is the generator of the space-time process $(t,\tilde{X}_t)$. If $\tilde\cQ = (\tilde q(x,y),\, x, y \in E)$ denotes the $Q$-matrix of $\tilde X$, then the integrand in the exponential of the weight in Theorem \ref{thm:continuous_time_weight} equals
\[ \int_0^T \frac{({\cL}-\tilde{{\cL}}) g}{g} (t, X^\circ_t) \dd t =  \int_0^T  \frac{ \sum_{y \in E} (q(X^\circ_{t},y)-\tilde{q}(X^\circ_{t},y)) g(t,y)}{g(t,X^\circ_t)} \dd t.\]

For finite-state Markov chains transition densities can be computed via $(p_t(x,y),\, x, y \in E) = e^{t \cQ}$ and this would immediately give $h$, implying no need for constructing a guided process via an approximation to Doob's $h$-transform. While strictly speaking this is true, in case $|E|$ is large, numerically approximating  the matrix exponential can be computationally demanding. In that case, by simplifying the dynamics of the Markov process, one may choose an approximating continuous time Markov process $\tilde X$ (on $E$) where $\tilde \cQ$ is block diagonal, as this simplifies evaluation of matrix exponentials.


\section{Sampling and statistical inference using guided processes}\label{sec:parestimation}
In this section we discuss a variety of ways in which Theorem \ref{thm:lr_xstar_xcirc_dag} can be used. 
Let $\P^\circ$ denote the law of $X^\circ_\cS$ and let $\P^\star$ denote the law of $X_{\cS} \mid X_\cV$. Then Theorem \ref{thm:lr_xstar_xcirc_dag} states
\[ \frac{\dd \P^\star}{\dd \P^\circ}(X^\circ_\cS) = \frac{g_0(x_0)}{\ev(x_0)} \Psi(X^\circ_\cS) \]
where \[ \Psi(X^\circ_\cS) = \prod_{s\in \cS}  \weight{\pa(s)}{s}{X^\circ_{\pa(s)}}  \prod_{v\in \cV} \frac{ h_{\pa(v) \pf v}(X^\circ_{\pa(v)})}{g_{\pa(v)\pf v}(X^\circ_{\pa(v)})} .
\]
Taking $f\equiv 1$ in its statement we obtain 
\begin{equation} \label{eq:evidence}\ev(x_0) = g_0(x_0) \E \Psi(X^\circ_\cS).\end{equation}
Our standing assumption is that we can define $g_{s\pf t}$ for all edges such that  $g_s$ can be computed efficiently for all $s\in \cS_0$. Depending on the  the dynamics of $X$ and $g$, it may then be possible to easily forward sample $X^\circ$ from root vertices towards leaf vertices (for example, in Gaussian nonlinear smoothing this is the case, cf.\ section \ref{subsec:gauss}). Then, again depending on $X$ and $g$ the weights $w_{\pf s}$  may be tractable or not. 

Note that $x_0$ may encode any unknown parameters such as the starting point of a stochastic process or a parameter $\theta$ reflecting uncertainty about the forward dynamics. Then formally, there is an edge from $x_0$ to every other vertex in the graph. In  the backward filtering step we are free to drop such edges. In that case, the evidence depends on $\theta$. While the backward filter does not return an approximation to the distribution of $\theta$ conditional on the leaf observations, the evidence can still be computed via \eqref{eq:evidence} and used to infer the parameter.

 There are links with a variety of  popular methods from computational statistics:

\subsection{Importance sampling}  Suppose interest lies in $I:=\E \phi(X^\star)$. We simulate independent copies $\{X^{\circ,i},\, i =1,\ldots B\}$ under $\P^\circ$. Next, define the importance sampling estimator 
\[ \hat{I}:= \sum_{i=1}^B w_i \phi\left(X^{\circ,i}\right), \qquad w_i = \frac{\Psi(X^{\circ,i})}{\sum_{i=1}^B \Psi(X^{\circ,i})}. \]
Note that $w_i$ does not depend on $\ev(x_0)$ and the estimator requires that the weights can be computed efficiently (if this is not the case, this can however be  addressed by  Pseudo-Marginal Metropolis-Hastings, see below). 

\subsection{Variational inference}	
If $g$ is indexed by a parameter $\eta \in D$, then it induces a  family of laws  
	 $\cP= \{\P^{\circ (\eta)}, \eta \in D\}$ which can be used as a  variational class  in variational inference. 
That is, we find an element in $\cP$ that  approximates $\P^\star$ by the information projection (an early reference in the context of guided processes is \cite{schauer2017guided}). More precisely, we aim to find $\argmin_{\eta \in D} \on{KL}(\P^{\circ (\eta)}, \P^\star)$, $\on{KL}$ denoting Kullback-Leibler divergence. The minimisation can be 
	done using stochastic gradient descent, which requires estimating $\nabla_\eta \on{KL}(\P^{\circ (\eta)}, \P^\star)$.
Typically, the proposal process can be reparametrised as 
\begin{equation}\label{eq:Xnormalisingflow}
X^{\circ(\eta)} = G(\eta, Z)
\end{equation}
for a deterministic map and a parameter-free innovation process $Z$ (for example white noise), say with law $\Q$. This enables variational inference with the ``reparametrisation trick'' \citep{kingma2013autoencoding}.  More precisely, 
define $\Q^{\star}$ by
\[
\frac{\dd\Q^{\star}}{\dd \Q}(Z) = \frac{\dd\P^\star}{\dd \P^\circ} (G(\eta,Z)).
\]
By the transformation formula, samples $Z \sim \Q^{\star}$ give $G(\eta,Z) \sim \P^\star$.  Since
\[ \nabla_\eta \E^{\circ(\eta)} \log \frac{\dd\P^\star}{\dd \P^{\circ(\eta)}}(X^{\circ(\eta)} ) =\nabla_\eta \E^{\Q} \log \frac{\dd\P^\star}{\dd \P^{\circ(\eta)}}(G(\eta,Z)) =  
\E^{\Q} \nabla_\eta \log \frac{\dd\P^\star}{\dd \P^{\circ(\eta)}}(G(\eta,Z))\]
an estimator for this gradient can be obtained by sampling from $\Q$ (we implicitly assume expectation and differentiation can be interchanged). 

Note that the construction of $X^{\circ(\eta)}$ in \eqref{eq:Xnormalisingflow} is a normalising flow, see for instance \cite{pmlr-v37-rezende15} and  \cite{Kobyzev2020}. Alternatively, it can be seen as a noncentred parametrisation \citep{PapasRobertsSkoeld} or 
randomness pushback \citep{Fritz2020}.

\subsection{Markov chain Monte Carlo methods}\label{subsec:mcmc}
Suppose we can reparametrise $X^\circ = g(Z)$ where $Z$ is standard normal. Then a Metropolis-Hastings algorithm for sampling from $\P^\star$ can be constructed by  sampling $W \sim N(0,I)$ (independently from $Z$), choosing $\alpha \in [0,1)$  and proposing
\[ Z_{\rm new} = \alpha Z +\sqrt{1-\alpha^2} W. \]
Next, $Z_{\rm new}$ is accepted with  probability \[ 1\wedge  \frac{\Psi(G(Z_{\rm new}))}{\Psi(G(Z))} \] and upon acceptance $X^\circ$ is updated to  $G(Z_{\rm new})$. The update step for $Z_{\rm new}$ is often referred to as a preconditioned Crank-Nicolson (pCN) step.

In evaluating  $\Psi(X^\circ_{\cS})$, the product of weights, $\prod_{s\in \cS}  \weight{}{s}{X^\circ_{\pa(s)}}$ may be intractable. The Pseudo-Marginal Metropolis-Hastings (PMMH) algorithm \citep{andrieu2009pseudo} is based on the observation that if $\Psi(X^\circ_{\cS})$ is replaced with an unbiased estimator, the Metropolis-Hastings algorithm still targets the exact target distribution. 
Here, this algorithm can be used upon realising that  $\Psi(X^\circ_\cS) $ can be unbiasedly estimated by 
 \[  
   \prod_{s\in \cS} 
\frac{g_s(X_s)}{\prod_{u\in \pa(s)} g_{u\pf s}(X^\circ_u) } 
   \prod_{v\in \cV} \frac{ h_{\pa(v) \pf v}(X^\circ_{\pa(v)})}{ g_{\pa(v)\pf v}(X^\circ_{\pa(v)})} .
\]
where $X_s \sim \scr{L}(X_s \mid X_{\pa(s)} = X^\circ_{\pa(s)})$.

\subsection{Hamiltonian Monte Carlo}
This is also helpful for derivative-based samplers such as Hamiltonian Monte Carlo.
Consider the situation where $X^\circ = G(Z)$, for $Z \sim \Q=\uN(0,\uI)$.%

The Hamiltonian to be used in a Hamiltonian Monte Carlo sampler targeting $Z^\star$ is  
\[
 \uH( \mathsf z, \mathsf p)= -\log \frac{\dd\Q^{\star}}{\dd \Q}\frac{\dd \Q}{\dd \lambda}( \mathsf z)  + \frac12 \| \mathsf  p\|,
\]
where $ \mathsf p$ is the momentum vector and $\lambda$ denotes the Lebesgue measure. 
Thus with 
\[
\psi_s( \mathsf  z) = -\pder{ \mathsf  z_s} \log \frac{\dd\Q^{\star}}{\dd \Q}\frac{\dd \Q}{\dd \lambda}( \mathsf  z) =  - \sum_{ t \in \cS} \pder{ \mathsf z_s} \log \weight{}{t}{G( \mathsf z_{\pa(t)}}  -  \mathsf z_s.
\]
the coordinate functions of the gradient of the negative score (the potential energy) and with $\lambda$ the Lebesgue measure, the Hamiltonian dynamics are given by
\[
\frac{\partial  \mathsf p_s}{\partial t} = -\frac{ \partial  \uH( \mathsf z,  \mathsf p)}{\partial  \mathsf z_s} = -\psi_s( \mathsf z) , \quad \frac{\partial  \mathsf z_s}{\partial t}  = \frac{ \partial  \uH( \mathsf z,  \mathsf p)}{\partial  \mathsf p_s} =  \mathsf p_s.
\]
Thus a Hamiltonian Monte Carlo method to sample $Z^\star$ can be used and samples of the target $X^\star$ can then obtained via transformation $G(Z^\star)$.


\bigskip

{\bfseries Acknowledgement:}   
Frank Schäfer implemented Mitosis rules for stochastic differential equation based on earlier code by Marcin Mider.  We thank Richard Kraaij for discussions on the topic of $h$-transforms in the early stage of this work. 
We thank for valuable input Marc Corstanje, Keno Fischer, Philipp Gabler, Evan Patterson and Chad Scherrer.

\small

\nocite{rackauckas2020universal}
\nocite{innes2018dont}
\nocite{lahoz2010data}
\nocite{carlson2020montecarlomeasurementsjl}

\bibliographystyle{harry}
\bibliography{literature}

  
\appendix

\section{Summary of related work}\label{sec:related work}
In case of a line graph with edges to observation leaves at each vertex (a hidden Markov model) there are two well known cases for computing the smoothing distribution: {\it (i)} if $X$ is discrete the forward-backward algorithm (\cite{Murphy2012}, section 17.4.3), {\it (ii)} for linear Gaussian systems the Kalman Smoother, also known as Rauch-Tung-Striebel smoother (\cite{Murphy2012}, section 18.3.2) for the marginal distributions or its sampling version, where samples from the smoothing distribution are obtained by the forward filtering, backward sampling algorithm, \cite{CarterKohn(1994)}.  \cite{Pearl1988} gave an extension of the forward-backward algorithm from chains to trees by an algorithm known as  ``belief propagation'' or ``sum-product message passing", either on trees or poly-trees. This algorithm  consists of two message passing phases. In the ``collect evidence'' phase, messages are sent from leaves to the root; the ``distribute evidence'' phase ensures updating of marginal probabilities or sampling joint probabilities from the root towards the leaves.
The algorithm can be applied to junction trees as well, and furthermore, the approximative loopy belief propagation applies belief propagation to sub-trees of the graph. A review is given in \cite{Jordan2004}. 

\cite{Chou1994} extended the classical Kalman smoothing algorithm for linear Gaussian systems to dyadic trees by using a fine-to-coarse filtering sweep followed by a course-to fine smoothing sweep. This setting arises as a special case of our framework. Extensions of filtering on  Triplet Markov Trees and pairwise Markov trees are dealt with in \cite{BardelDesbouvries2012}  and \cite{DesbouvriesLecomtePieczynski} respectively. 

Particle filters have been employed in related settings, see \cite{doucet2018sequential} for an overview, but the resampling operations of particle filters 
result in the simulated likelihood function being non-differentiable in the parameters, which unlike our approach is an obstacle to gradient based inference. 
\cite{briers2010smoothing} consider particle methods for  state-space models using (an approximation to) the backward-information filter. Twisted particle samplers (\cite{Guarniero2017}) and controlled Sequential Monte Carlo (SMC) (\cite{Heng2020}) essentially use particle methods to find an optimal control policy to approximate the backward-information filter. The approximate backward filtering step in this paper uses a similar idea, formulated in a broader context than particle filtering and discrete-time models. 

To easily compare to \cite{Heng2020} (shortly referred to as [H] here), we assist in bridging notation used in that paper and ours. The optimal Backward Information Filter (Cf.\ Section \ref{subsec: doob}), which we denote by $h_t$ at node $t$, is denoted $\psi^\star_t$ in [H]. Approximations to $h_t$ are specified in our setup by providing $g_{s\pf t}$ over all edges (leading to $\{g_s,\, s\in \cS\}$ by fusion). An a state-space model, those approximations  are  denoted by $\psi_t$ in [H].  The set of $\psi_t$ functions over all non-leaf vertices is then called a policy. For a state-space model it is shown how  $g_{\pa(v)\pf v}$ can be chosen such that the smoothing distribution is exact (this can also easily be deduced from Theorem \ref{thm:lr_xstar_xcirc_dag}). While not the focus of this paper, the idea of policy refinement outlined in [H] is also applicable to our setup.

For variational inference, \cite{ambrogioni2021automatic} propose an approach which, similar to ours, preserves the Markovian structure of the target  by learning local approximations to the conditional dynamics.

\section{Graph theory}\label{sec:notation}

Let $\cG = (\cT, \cE)$ be a directed, finite  graph without parallel edges and self-loops, with vertices  $\cT$ and edges $\cE$, where $e = (s, t)$, $s, t \in \cT$ denotes a directed edge from $s$ to $t$. 
\begin{itemize}
	\item The parent set of a vertex is the set of all vertices that feed into it: $\pa(s) =\{t\colon e(t,s) \in \cE\}$. 
	\item The children set of a vertex is the set of all vertices that feed out of it: $\ch(s)= \{t\colon e(s,t) \in \cE\}$.
\item Define $s \before t$ recursively by $s = t$ or $s \before a$ with $(s, a) \in \cE$ for some $a$. That is, $s = t$ or there is a directed path from $s$ to $t$. 
	\item Ancestors of a vertex are its parent, grand-parents, etc.: $\anc(t) = \{s \ne t\colon s \before t\}$.
	 Thus, for fixed $t$, $\{s\colon s \before t\} = \anc(t)\cup\{t\}$. 

\item A leaf (or sink) is a vertex with no children. The set of leaves is denoted by $\cV$. 
\item For each vertex $t$, we denote by $\cV_t = \{v \in \cV, t \before v\}$ the leaves which are descendants of $t$ (including $t$, if $t$ is a leaf).
\item A source is a vertex with no parents. 
	\end{itemize}


\section{Proofs}\label{sec:proofs}

\begin{proof}[Proof of Proposition \ref{prop: kl_backward}]
We aim to find $g_i$ minimising $-\sum_{i=1}^d \int \pi(x) g(x)  \log g_i(x_i) \dd x$. Writing $\log g_i = \tilde g_i$ and introducing Lagrange-multiplier we consider
\[ \cL(\tilde g_1,\ldots, \tilde g_d, \lambda_1,\ldots, \lambda_d) = -\sum_{i=1}^d \pi(x) g(x) \tilde g_i(x_i) \dd x +\sum_{i=1}^d \lambda_i \left(\int e^{\tilde g_i(x_i)} \pi_i(x_i)-c_i\right). \]
Now consider $\cL$ with its $j$-th argument taken to be $\tilde g_j(x_j) + \epsilon \delta(x_j)$, subtract $\cL$ and collect terms of $\scr{O}(\epsilon)$ to get
\[ -\int \pi(x) g(x) \delta(x_j) \dd x + \lambda_j \int g_j(x_j) \pi_j(x_j) \delta(x_j) \dd x_j. \]
Equating this to zero gives
\[ \int \delta(x_j) \left( \lambda_j g_j(x_j) \pi_j(x_j)  -\int \pi(x) g(x) \dd x_{-j}\right) \dd x_j = 0\]
which implies 
\[  \lambda_j g_j(x_j) \pi_j(x_j)  =\int \pi(x) g(x) \dd x_{-j}. \]
Integrating over $x_j$ gives $\lambda_j c_j = 1$ from which the result easily follows.
\end{proof}

\begin{proof}[Proof of Theorem \ref{thm:gaussianformulas}]
First note that 
\begin{align*} (\kappa g)(x) &= \int g(y) \kappa(x,\dd y) =\int \varpi(c,F,H)\phi(y; H^{-1}F, H^{-1}) \phi(y; \mu(x), Q(x)) \dd y
 \\ &= \varpi(c, F, H) \phi(H^{-1}F; \mu(x), Q(x) + H^{-1}).
\end{align*}
By using the specific form of $\mu(x)$ and $Q(x)$ for the kernel $\tilde\kappa$ the expression for the weight follows.

To find the parametrisation of $\tilde\kappa h$, let $C=Q+H^{-1}$. We have
\begin{align*}
(\tilde \kappa g)(x) &= \varpi(c, F, H) \phi(H^{-1}F; \Phi x + \beta, Q + H^{-1}) \\ 
& =\varpi(c, F, H) (2\pi)^{-d/2} |C|^{-1/2} \\ & \quad \times \exp\left(-\frac12 x^\T \Phi^\T C^{-1} \Phi x + (H^{-1}F-\beta)^\T C^{-1} \Phi x - \frac12 (H^{-1}F-\beta)^\T C^{-1} (H^{-1}F-\beta)\right)\\  & = U(c - \log \phi^{\on{can}}(0, F, H) +\log \phi(\beta; H^{-1}F, C),  \Phi^\T C^{-1}(H^{-1}F-\beta), \Phi^\T C^{-1} \Phi) (x). 
\end{align*}
By equating this to $U(\bar c, \bar F, \bar H)$ the parametrisation can be inferred directly,
If $\Phi$ is invertible, then 
\[ \bar F^\T \bar H \bar F = (H^{-1} F -\beta)^\T C^{-1} (H^{-1} F -\beta). \]
A bit of algebra then gives the stated result.

The derivation of the parametrisation of the fusion step is trivial. 

 To derive the forward map, we write $\propto$ to denote proportionality with respect to $y$
\begin{align*} \frac{g(y) \kappa(x,\dd y)}{\dd y} &\propto \exp\left(-\frac12 y^\T H y + y^\T F\right) \exp\left(-\frac12 (y-\mu(x))^\T Q(x)^{-1} (y-\mu(x))\right) \\ &\propto \exp\left(-\frac12 y^\T (H+ Q(x)^{-1}) y+ y^\T (F+Q(x)^{-1}\mu(x))\right)\\ & \propto \phi^{\mathrm{can}}(y; F + Q(x)^{-1} \mu(x),H + Q(x)^{-1})
\end{align*}
which suffices to be shown.
\end{proof}

\begin{proof}[Proof of Theorem \ref{lem:gamma_incr_htilde_update}]
For deriving the pullback of $\kappa$, we compute
\begin{align*}
(\kappa g)(x) & = \int_x^{x_v} \psi(x_v-y; A, \beta) \psi(y-x; \alpha, \beta(x)) \dd y \\ & = \frac{\beta^A \beta(x)^\alpha}{\Gamma(A) \Gamma(\alpha)} e^{-\beta x_v +\beta(x) x} \int_x^{x_v} (x_v-y)^{A-1} (y-x)^{\alpha-1} e^{(\beta-\beta(x)) y } \dd y \\ & = 	
\frac{\beta^A \beta(x)^\alpha}{\Gamma(A) \Gamma(\alpha)} e^{-\beta (x_v-x)} (x_v-x)^{A+\alpha-1} \int_0^1 z^{\alpha-1} (1-z)^{A-1} e^{-z\xi(x)} \dd z 
\\ & = g_{A+\alpha, \beta}(x) \frac1{B(\alpha, \beta)} \int_0^1 z^{\alpha-1} (1-z)^{A-1} e^{-z\xi(x)} \dd z 
\end{align*}
where we made the substitution $y=x+z(x_v-x)$ at the third equality and $\mathrm{B}(\alpha, \beta)$-denotes the Beta-function, evaluated at $(\alpha, \beta)$.

From this expression the pullback for $\tilde\kappa$ follows directly, since  $\xi(x)=\xi$ implies $\xi(x)=0$. This also directly gives the expression for the weight. 

If $\mu=\delta_x$, then drawing from the forward measure entails drawing from a density which is proportional to $g(y) \kappa(x, \dd y)$ (proportionality with respect to $y$). The claim now follows upon inspecting the derivation of $(\kappa g)(x)$ and keeping all terms under the integral proportional to $z$.

\end{proof}

\end{document}